\documentclass[10pt,prl,prl,twocolumn,superscriptaddress,floatfix]{revtex4-2}
\usepackage{amsmath,amssymb,amsfonts,relsize,mathtools,times}%
\usepackage{amsthm}%
\usepackage{mathrsfs}%
\usepackage[title]{appendix}%
\usepackage{xcolor}%
\usepackage{textcomp}%
\usepackage{subfigure}

\usepackage{booktabs}%
\usepackage{algorithm}%
\usepackage{algorithmicx}%
\usepackage{algpseudocode}%
\usepackage{listings}%
\usepackage{bm}
\usepackage{bbold}
\usepackage{color}
\usepackage[english]{babel}
\usepackage{lipsum}
\usepackage[normalem]{ulem}
\usepackage{grffile}
\usepackage{braket}
\usepackage{appendix}
\usepackage[colorlinks, linkcolor=blue, anchorcolor=blue, citecolor=blue, urlcolor=blue]{hyperref}
\newcommand{\tr}[1]{\mathrm{tr}\!\left\{#1\right\}}

\newcommand{\bfu}{\mathbf{u}}
\newcommand{\bfx}{\mathbf{x}}
\newcommand{\bfv}{\mathbf{v}}

\newcommand{\bfc}{\mathbf{c}}

\newcommand{\bfa}{\mathbf{a}}
\newcommand{\bfb}{\mathbf{b}}

\newcommand{\bfy}{\mathbf{y}}
\newcommand{\bfZ}{\mathbb{Z}}
\newcommand{\bfF}{\mathbb{F}}

\newcommand{\qwcgroup}{\mathfrak{S}_n^{\rm qwc}}
\newcommand{\poly}{\mathrm{poly}}
\newcommand{\sgn}{\mathrm{sgn}}

\newcommand{\comb}[2]{{C}^{#2}_{#1}}
\newcommand{\hmoment}{  \mathbb{E}_{\psi\in {\rm Haar}}}
\newcommand{\linearstructure}{\rm LS}

\newcommand{\rank}{{\rm rank}}
\newtheorem{corollary}{Corollary}

\newtheorem{lemma}{Lemma}
\newtheorem{theorem}{Theorem}
\newtheorem{proposition}{Proposition}%
\newtheorem{definition}{Definition}%

\begin{document}

\title{Sample- and Hardware-Efficient Fidelity Estimation by Stripping Phase-Dominated Magic}

\author{Guedong Park}
\affiliation{NextQuantum Innovation Research Center, Department of Physics and Astronomy, Seoul National University, Seoul, South Korea}

\author{Jaekwon Chang}
\affiliation{Department of Physics, Korea University, Seoul, South Korea}

\author{Yosep Kim}
\affiliation{Department of Physics, Korea University, Seoul, South Korea}

\author{Yong Siah Teo}
\email{yong.siah.teo@gmail.com}
\affiliation{NextQuantum Innovation Research Center, Department of Physics and Astronomy, Seoul National University, Seoul, South Korea}
\affiliation{Department of Quantum Information Science and Engineering, Sejong University, Seoul, South Korea}

\author{Hyunseok Jeong}%
\email{h.jeong37@gmail.com}
\affiliation{NextQuantum Innovation Research Center, Department of Physics and Astronomy, Seoul National University, Seoul, South Korea}

\begin{abstract}
Direct fidelity estimation (DFE) is a famous tool for estimating the fidelity with a target pure state. However, such a method generally requires exponentially many sampling copies due to the large magic of the target state. This work proposes a sample- and gate-efficient fidelity estimation algorithm that is affordable within feasible quantum devices. We show that the fidelity estimation with pure states close to the structure of phase states, for which sample-efficient DFE is limited by their strong entanglement and magic, can be done by using $\mathcal{O}(\poly(n))$ sampling copies, with a \emph{single} $n$-qubit fan-out gate. As the target state becomes a phase state, the sampling complexity reaches $\mathcal{O}(1)$. Such a drastic improvement stems from a crucial step in our scheme, the so-called phase stripping, which can significantly reduce the target-state magic. Furthermore, we convert a complex diagonal gate resource, which is needed to design a phase-stripping-adapted algorithm, into nonlinear classical post-processing of Pauli measurements so that we only require a single fan-out gate. Additionally, as another variant using the nonlinear post-processing, we propose a nonlinear extension of the conventional DFE scheme. Here, the sampling reduction compared to DFE is also guaranteed, while preserving the Pauli measurement as the only circuit resource. We expect our work to contribute to establishing noise-resilient quantum algorithms by enabling a significant reduction in sampling overhead for fidelity estimation under the restricted gate resources, and ultimately to clarifying a fundamental gap between the resource overhead required to understand complex physical properties and that required to generate them.

\end{abstract}

\maketitle

\emph{Introduction.---}Quantum algorithms~\cite{google2023, zhang2025, shi2022} solve various challenging problems in modern quantum physics~\cite{hu2025_ansatz,huangH2023,zhang2025,banerjee2020,yin2020}. A key requirement of its implementation is the preparation of highly entangled pure states~\cite{vidal2003,miller2016,jhuang2024,howard2017,bravyi2019,mari2012}, which serve as essential quantum resources. However, this task inevitably introduces physical noise~\cite{google2023,harper2020}. Consequently, \emph{fidelity estimation}~\cite{flammia2011,cerezo2020,zhang2021} (FE) between a prepared state and its target pure state with efficient computational resources is an indispensable method for developing improved state-preparation protocols~\cite{zhang2022_qsp,wang2025} and for the reliable implementation of quantum algorithms.

One of the most representative protocols for fidelity estimation is the \emph{direct fidelity estimation} (DFE) scheme~\cite{flammia2011}, which requires only Pauli measurements as circuit resources. However, its sampling and time complexity typically grow exponentially with the number of qubits $n$~\cite{fawzi2024,liu2025_fe,julia2025}. State verification (or certification)~\cite{huang2024_sv,zhu2019,sun2025} is more likely to be executed efficiently using Pauli measurements, but it gives a possible fidelity interval~\cite{sater2025,zhu2019,huang2024_sv}, not the exact value~\cite{sun2025}. 
Recent approaches, such as machine learning~\cite{zhang2021,cerezo2020,Vadali2023}, quantum phase estimation~\cite{wang2024_sample_optimal}, classical-shadows-assisted~\cite{aronson2018,huang2020,schuster2025random} FE have achieved an improvement in sampling copies. Nevertheless, there are limitations in that they require costly gate resources~\cite{cerezo2020,wang2024_sample_optimal}. For example, $\mathcal{O}(n^2)$ gates are needed for unbiased classical-shadow tomography~\cite{bravyi2021,huang2020}, which may not be achievable in near-term platforms~\cite{song2025,zheng2025}. Likewise, there exists an inherent trade-off between sampling complexity and gate complexity. Therefore, establishing the fundamental limits of this trade-off, namely, optimizing sampling under restricted gate resources, remains a key open challenge.

In this work, we propose sample- and gate-efficient FE techniques that utilize classical yet \emph{nonlinear} post-processing of the Pauli measurement outcomes. 
By appending a single ancilla and employing a fan-out gate~\cite{song2025} implemented with $n$ CNOTs, our method significantly reduces the sampling complexity of FE.
The fan-out gate has been demonstrated in current experimental platforms~\cite{song2025,lu2019} and suffices for our implementation. 

Our method is motivated by two key observations.
First, the sampling inefficiency of DFE is derived from the exponentially large Pauli $l_{(1\;{\rm or}\;0)}$-norm~\cite{chen2024_magic_mps} ($l$-norm, shortly), which is a typical measure of magic~\cite{leone2023,chen2024,howard2017} of the target state. Second, the \emph{phase stripping}, which transforms all coefficients of the computational bases to their modulus, can significantly reduce the $l$-norm. This phenomenon becomes pronounced when most of the target state magic is driven by complex diagonal gates~\cite{gosset2025} (\emph{phase-dominated magic}). Accordingly, we shall design an FE algorithm whose sampling complexity depends on the $l$-norm  of the phase-stripped version of the target state. After that, we see that the required sampling complexity of our FE algorithm is drastically reduced for some cases. A significant and practical example is the phase state~\cite{huang2024_sv,arunachalam2023,lee2025_shallow,morimae2017} for which phase stripping releases all magic and compresses into the Hadamard basis, reducing the sampling complexity down to $\mathcal{O}(1)$. 

To this end, we build on the Hadamard test circuit~\cite{sun2022,tsubouchi2023_virtual,faehrmann2025}, which will be shown to estimate the target fidelity. After that, complex controlled-diagonal operations are translated into a nonlinear post-processing of Pauli measurement outcomes, leaving the fan-out gate as the only \emph{physical} entangling resource. An important remark is that the $l$-norm of the phase-stripped state of arbitrary phase states~\cite{lee2025_shallow,huang2024_sv} is unity. It enables our scheme to require only $\mathcal{O}(1)$-sampling copies for the phase state FE, whereas DFE still needs $\mathcal{O}(2^n)$ copies. We then discuss cases in which our scheme is also time- and memory-efficient, including phase states. Finally, we exhibit numerical simulations to validate the accuracy and efficiency of our scheme. We first compare with DFE~\cite{julia2025}, or alternative FEs such as shadow overlap~\cite{huang2024_sv}, and Clifford shadow~\cite{huang2020,west2025}.

Additionally, we propose a nonlinearly-augmented~DFE framework catered to the situation where \emph{only} Pauli measurements are allowed. Phase stripping is not applied here, but to reduce the sampling overhead, we employ a nonlinear post-processing of the direct Pauli measurements. 
We give a systematic algorithm for the \emph{nonlinear}-DFE, based on the divide-and-conquer (DNC) strategy. While DNC requires exponential memory and time by the number of qubits, it does not require the convex optimization~\cite{cha2025_operator_aware,caprotti2026} over such memory data or input states.

\textit{Direct FE.}---We first outline the conventional algorithm for estimating the fidelity between an $n$-qubit input state $\rho$ and the target pure state $\ket{\psi}\bra{\psi}$, or the direct FE (DFE) scheme~\cite{flammia2011,julia2025,leone2023}. For the $n$-qubit Pauli group as
$\mathcal{P}_n\equiv \left\{\pm iI,\pm iX,\pm iY,\pm iZ\right\}^{\otimes n}$~\cite{gottesman1998} and a given single-qubit Pauli operator $P$, let us denote an $n$-qubit Pauli operator $P^{\bfa}\equiv \bigotimes_{i=1}^{n} P^{a_i}$, where $\bfa\in \bfF^n_2$.
It follows that an up-to-phase Pauli operator in $\mathcal{P}_n/ \mathbb{Z}_4$ is $T_{\bfa}\equiv \bigotimes_{i=1}^{n}i^{a_{ix}a_{iz}}X^{a_{ix}}Z^{a_{iz}}$, where $\bfa=(\bfa_x,\bfa_z)\in \bfF^{2n}_2$. From the general decomposition using Pauli coefficients,
\begin{align}\label{eq:Pauli_expansion}
    \ket{\psi}\bra{\psi}=\sum_{\bfa\in \bfF_2^{2n}}c_{\psi}(\bfa)T_{\bfa}\;\;\left(c_{\psi}(\bfa)\equiv \frac{1}{2^n}\braket{\psi|T_{\bfa}|\psi}\right),
\end{align}
for a pure state $\ket{\psi}$ we define the Pauli rank~\cite{bu2019} (or $l_0$-norm) $\|\psi\|_0\equiv\frac{1}{2^n}\sharp\left\{\bfa\in \bfF_2^{2n}|c_{\psi}(\bfa)\ne 0\right\}$, and Pauli $l_{1}$-norm (or $l_1$-norm)~\cite{chen2024_magic_mps,howard2017,leone2022} $\|\psi\|_1\equiv \frac{1}{2^n}\sum_{\bfa\in \bfF_2^{2n}}|\braket{\psi|T_{\bfa}|\psi}|\ge \|\psi\|_2\equiv 1$. For $\alpha\in \left\{\frac{1}{2},1\right\}$, it follows that $\|\chi\|_{2-2\alpha}=1$ for all pure stabilizer states $\ket{\chi}$~\cite{aaronson2004}. Using the above decomposition, we design an $\alpha$-DFE scheme utilizing Pauli measurements as follows~\cite{leone2023}: first sample $\bfa\in \bfF_2^{2n}$ following the $l_{2\alpha}$-distribution, $\left\{ \frac{2^{(2\alpha-1)n}|c_{\psi}(\bfa)|^{2\alpha}}{\|\psi\|_{2\alpha}}\right\}_{\bfa\in \bfF_2^{2n}}$ ($l_{2\alpha}$-\emph{sampling}), then measure $\rho$ to estimate  $2^{(1-2\alpha)n}\|\psi\|_{2\alpha}\tr{\rho T_{\bfa}}\sgn(c_{\psi}(\bfa))|c_{\psi}(\bfa)|^{-2\alpha+1}$. Upon repeating the procedure with $N\gg1$ copies, the mean of the estimates converges to the target value $\braket{\psi|\rho|\psi}$. The original version~\cite{flammia2011,julia2025} assumed $\alpha=1$. See Supplemental Material~\cite{supple} (SM) for details.

This work considers only when $\alpha\in \left\{\frac{1}{2},1\right\}$. The estimation variance of DFE, which determines the required sampling copy numbers for an estimation accuracy within an additive error $\epsilon$, is $\mathcal{O}(\|\psi\|^{1/\alpha}_{2-2\alpha})$~\footnote{If we allow arbitrary $\alpha\in (0,1)$, the variance is not $\mathcal{O}(\|\psi\|^{1/\alpha}_{2-2\alpha})$, but still the minimum is achieved when $\alpha=\frac{1}{2}$. See SM~\cite{supple} for details.}, where the minimal bound is achieved when $\alpha=\frac{1}{2}$~\cite{supple}.  DFE offers the simplest Pauli measurement-based quantum algorithm for FE. However, since $l_1$-norm of most pure states is exponentially large in~$n$~\cite{chen2024,chen2024_magic_mps}, the required sampling and time are known to be inefficient~\cite{julia2025}.
Representative examples of inefficiency for DFE, such as the \emph{phase states}, will be introduced in the next section. These states form our main study targets and highlight the significance of our results.

\textit{Phase stripping and phase states}---
In general, an arbitrary pure state $\ket{\psi}$ is the output of some diagonal gate operation $D(\phi)_{\psi}\equiv \sum_{\bfx\in \bfF_2^n}e^{i\phi_{\psi}(\bfx)}\ket{\bfx}\bra{\bfx}$, with some phase function $\phi_{\psi}:\bfF_2^n\rightarrow[0,2\pi]\;(\in [0,2\pi]^{2^n})$, on its real-valued counterpart. That is,
\begin{align}\label{eq:phase_stripping}
    \ket{\psi}=\sum_{\bfx\in \bfF_2^n}\xi_{\bfx}\ket{\bfx}=\sum_{\bfx\in \bfF_2^n}e^{i{\rm arg}(\xi_{\bfx})}|\xi_{\bfx}|\ket{\bfx}=D(\phi_{\psi})\ket{\breve{\psi}},
\end{align}
where $\forall \bfx\in \bfF_2^n,\;\phi_{\psi}(\bfx)\mapsto \arg(\xi_{\bfx})\in [0,2\pi]$, and the \emph{phase-stripped state} $\ket{\breve{\psi}}\equiv \sum_{\bfx\in \bfF_2^n}|\xi_{\bfx}|\ket{\bfx}$. We call such mapping $\ket{\psi}\mapsto \ket{\breve{\psi}}$ as \emph{phase stripping}. For brevity, we denote $\breve{c}_{\bfa}=c_{\breve{\psi}}(\bfa)$.

While fixing the DFE index $\alpha\in \left\{\frac{1}{2},1\right\}$, we call $\ket{\psi}$ a \emph{phase state}~\cite{lee2025_shallow,huang2024_sv} if and only if $\ket{\breve{\psi}}=\ket{+}^{\otimes n}\;\left(\ket{+}\equiv\frac{\ket{0}+\ket{1}}{\sqrt{2}}\right)$. If a phase state $\ket{\psi}$ satisfies $\forall \bfx\in \bf\bfF_2^n,\;\phi_{\psi}(\bfx)\in \left\{0,\pi\right\}$, we call it \emph{hypergraph state}~\cite{zhu2019,morimae2017}. Phase states serve as fundamental resource states for universal quantum computing~\cite{briegel2009,miller2016,lee2024,morimae2012}, quantum cryptography~\cite{arunachalam2023,banerjee2020}, and pseudo-randomness~\cite{brakerski2019,lee2025_shallow}. Hence, sample- and time-efficient phase-state FE is a crucial task for the practical realization of quantum simulations. However, DFE cannot achieve such efficiency for phase states. For example, we can show that if the target phase state $\ket{\eta}$ is a random third-ordered hypergraph state~\cite{zhu2019, chen2024}, $\|\eta\|_1\simeq\Theta(2^{\frac{n}{2}})$~\cite{supple}. For this case, the sampling complexity for DFE is $\Theta(\|\eta\|_1^2)\simeq\mathcal{O}(2^n)$.

We have seen that the sample-inefficiency of DFE is closely connected to the exponentially large $l_{2-2\alpha}$-norm. Whereas the phase stripping could be a key to resolving such a problem. Indeed, these kinds of norms have a direct relation with the stabilizer Rényi entropy~\cite{leone2022,leone2023}, a typical magic measure of the pure state~\cite{ding2025}. That means, the 
phase stripping can significantly reduce the $l_{2-2\alpha}$-norm of a given target state, especially when most of the target state's magic is induced by $D(\phi_\psi)$---that is, \emph{phase-dominated magic}--- which is generated by using a $\mathcal{O}(2^n)$ Clifford$+$T gates~\cite{gosset2025}. As a concrete example,  recall that the phase stripping maps an arbitrary phase state to $\ket{+}^{\otimes n}$ whose $l_{2-2\alpha}$-norm equals $1$, whereas the original phase state typically exhibits an exponentially large  $l_{2-2\alpha}$-norm. Moreover, many pure states lie in \emph{near-phase} class, for which $\|\breve{\psi}\|_{2-2\alpha}=\mathcal{O}(\poly(n))$~\cite{bravyi2016i}. Therefore, we are naturally motivated by the following question: 
Can we design the phase-stripping-adapted FE algorithm for which the estimation variance depends on $\|\breve{\psi}\|^{1/\alpha}_{2-2\alpha}$? We will show in the next section that this is indeed possible.

\textit{Fan-out-based FE (FOFE).---}We have learned that DFE for target phase states requires $\mathcal{O}(2^n)$-sampling copies~\cite{flammia2011,fawzi2024}. Naive measurement with respect to target state bases requires only $\mathcal{O}(1)$-sampling copies, but requires $\mathcal{O}(2^n)$-gate complexity. Our result argues that only a single $n$-qubit fan-out gate ($n$ CNOTs with a common control qubit~\cite{song2025}) is sufficient to achieve the $\mathcal{O}(\|\breve{\psi}\|_{2-2\alpha}^{1/\alpha})$-sampling complexity, and such a gate has been realized in trapped-ion~\cite{lu2019} and superconducting-qubit systems~\cite{song2025}. Our result yields a corollary that the sampling copy complexity reduces to $\mathcal{O}(1)$ for target phase states. We first give the formal statement:
\begin{theorem}\label{main:thm_1}
    With the fixed $\alpha\in \left\{\frac{1}{2},1\right\}$, suppose we have $M$ different $n$-qubit target states $\left\{\ket{\psi_1},\ket{\psi_2},\ldots,\ket{\psi_M}\right\}$ such that all elements share the same phase-stripped state $\ket{\breve{\psi}}$. We also assume that the $l_{2\alpha}$-sampling of $\ket{\breve{\psi}}$, and the calculation of phase function value $\phi_{\psi_i}(\bfx)\;(i\in[M])$ for each $\bfx$ and $i$ can be efficiently and classically simulated.
    With a single-qubit ancilla state, at most $n$~CNOTs, and $(H,S)$ gates, the estimation of the fidelity between an input state $\rho$ and $\ket{\psi}$ up to an accuracy within an additive error $\epsilon>0$ and failure probability $\delta_f>0$ can be achieved with sampling complexity  $\mathcal{O}\left(\frac{\|\breve{\psi}\|_{2-2\alpha}^{1/\alpha}}{\epsilon^2}\log(M\delta_f^{-1})\right)$. As a corollary, fidelity estimation with the $M$ different target phase states can be done in $\mathcal{O}(\epsilon^{-2}\log(M\delta_f^{-1}))$-samples.
\end{theorem}

We call our scheme $\alpha$-\emph{fan-out-based FE} ($\alpha$-FOFE). As the name suggests, the scheme leverages the fan-out-based quantum algorithm, which is illustrated in  Fig.~\ref{fig:phase_state_fidelity_scheme}. Therefore, a single fan-out gate suffices to achieve sample-efficient FE for the \emph{near-phase} cases [$\|\breve{\psi}\|_{2-2\alpha}=\mathcal{O}(\poly(n))$ holds for at least one of $\alpha\in \left\{\frac{1}{2},1\right\}$], and manifests a significant sample-improvement as $\ket{\psi}$ shrinks to a phase state. We see that it neither requires prior block-diagonalization of the input state nor multi-copy measurements~\cite{ lees2025, tsubouchi2024,king2025}. 

The assumption of the classical simulation holds in famous cases, such as hypergraph states with bounded order~\cite{rossi2013,miller2016} or Dicke states~\cite{Brtschi2019,flammia2011,lucke2014} twirled by some diagonal operation $\ket{\psi}=D(\phi)\ket{{\rm Dic}(n,k)}$~\cite{supple}, which is also near-phase. For phase states, $\ket{\breve{\psi}}=\ket{+}^{\otimes n}$, and the $l_{2\alpha}$-sampling reduces to uniform sampling of Pauli $X$-operators. Moreover, the assumption exists only for the efficient time complexity and is independent of the \emph{sampling-copy} complexity. 

In general, $\frac{1}{2}$-FOFE offers tighter sampling bound than $1$-FOFE. Nevertheless, the pure state regime of efficient $l_2$ sampling in $1$-FOFE is much broader than $l_1$-sampling. In particular, the $l_2$-sampling is equivalent to the Bell sampling~\cite{montanaro2017_ls,supple}, where we enact the transversal $CNOT$ gates to $\ket{\breve{\psi}}^{\otimes 2}$ and take the Pauli measurements whose outcome exactly (up to a fixed permutation) follows the $l_2$-sampling distribution. Following that, if $\ket{\breve{\psi}}$ is generated by a large number of Clifford gates with small magic gates, we can apply various classical simulation algorithms~\cite{bravyi2016i,bravyi2019,qassim2021,howard2017} for a faster $l_2$-sampling. As another example, we can represent $\ket{\breve{\psi}}$ as a matrix product state (MPS)~\cite{haug2023_mps_nonstab,vidal2003} whose bond dimension~\cite{vidal2003} is $\chi(\breve{\psi})$. Regardless of the structural complexity of $D(\phi_\psi)$, we can classically simulate the $l_2$-sampling of $1$-FOFE within the computational complexity $\mathcal{O}(n\chi(\breve{\psi})^3)$~\cite{lami2023,tarabunga2024} time for each Pauli sample. 

\begin{figure}[t]
    \includegraphics[width=\columnwidth]{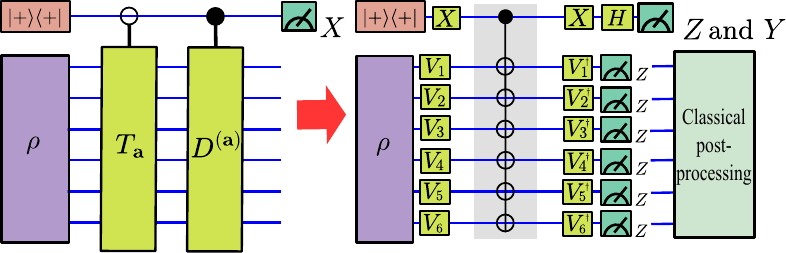}
    \caption{Schematic illustration of the $6$-qubit fan-out-based fidelity estimation (FOFE). We assume that the sampled $T_{\bfa}$ is of full Pauli weight.  Here, $\ket{+}$ is an ancilla state and $\rho$ is an input. The conjugation of single-qubit Clifford $V_i\;(i\in [6])$ is such that $VXV^{\dag}=T_{\bfa_i}$. }\label{fig:phase_state_fidelity_scheme}
    
\end{figure}

Now, we give an outline of the proof and algorithm constituting Thm.~\ref{main:thm_1}. Though we fix $\alpha=\frac{1}{2}$, the proof for $\alpha=1$ follows similarly. Complete explanations are shown in SM~\cite{supple}. Given an arbitrary state $\ket{\psi}$, we give the decomposition $\ket{\psi}=D(\phi_{\psi})\ket{\breve{\psi}}$ [see Eq.~\eqref{eq:phase_stripping}], where $\ket{\breve{\psi}}\bra{\breve{\psi}}=\sum_{\bfa\in \bfF_2^{2n}}\breve{c}_{\bfa}T_{\bfa}$ $(\forall\breve{c}_\bfa\in \mathbb{R})$ that leads to

\begin{align}
\ket{\psi}\bra{\psi}=\sum_{\bfa\in \bfF_2^{2n}}\breve{c}_{\bfa}D(\phi_\psi)T_{\bfa}D(\phi_\psi)^{\dag}=\sum_{\bfa\in \bfF_2^{2n}}\breve{c}_{\bfa}T_{\bfa}D^{(\bfa)},
\end{align}
and the simplified notation $D(\phi_{\psi})=D$. We also used the fact that  $D^{(\bfa)}\equiv T_{\bfa}DT_{\bfa}D^{\dag}$ is again another diagonal gate that depends on $\bfa\in \bfF_2^{2n}$. From this knowledge, similar to the original DFE, we now establish our enhanced FE scheme with the input $\rho$ as follows: (i) Sample $\bfa$ following the distribution $\left\{\frac{|\breve{c}_{\psi}(\bfa)|}{\|\breve{\psi}\|_1}\right\}_{\bfa\in \bfF_2^{2n}}$, where $\|\breve{\psi}\|_1=\sum_{\bfa\in \bfF_2^{2n}}|\breve{c}_{\bfa}|$. (ii) Compute the unbiased estimate 
\begin{align}\label{eq:main_phase_fidelity_estimator}
    \widehat{\braket{\psi}}=\frac{\|\breve{\psi}\|_1}{2}\sgn(\breve{c}_\bfa)\left[\tr{D^{(\bfa)}\rho T_{\bfa}}+\tr{T_{\bfa}\rho D^{(\bfa)\dag}}\right].
\end{align}

We employ the Hadamard test circuit~\cite{sun2022,tsubouchi2023_virtual,faehrmann2025}, which enables us to sample-efficiently estimate Eq.~\eqref{eq:main_phase_fidelity_estimator} excepting $\|\breve{\psi}\|_1\sgn(\breve{c}_{\bfa})$  (the left of Fig.~\ref{fig:phase_state_fidelity_scheme}). We thus expect the sampling complexity to be proportional to $\|\breve{\psi}\|_1^2$~\cite{Jerrum:1986random}, implying that our scheme for target phase state, say  $\ket{\eta}$, is sample-optimal since $\|\breve{\eta}\|_1=\|+\|_1^n=1$.  The main challenge is that the implementation of a controlled-diagonal unitary incurs substantially higher cost~\cite{gosset2025} compared to that of the earlier controlled-Pauli operator. To overcome this problem, we do not use the second controlled unitary. Instead, the Pauli measurement is performed independently on two copies (right after the controlled-Pauli operation) to obtain the binary outcome $\bfb=(b_1,\bfb')\in \bfF_2^{n+1}$ (the right of Fig.~\ref{fig:phase_state_fidelity_scheme}). This outcome is then nonlinearly post-processed to estimate two expectations, $\braket{Z\otimes {\rm Re}(D^{(\bfa)})}$ and  $\braket{-Y\otimes {\rm Im}(D^{(\bfa)})}$. Specifically, in the former case, the estimator will be $(-1)^{b_1}\cos(\phi^{(\bfa)}_{\psi}(\bfb'))$ after the measurement, where $\phi^{(\bfa)}_{\psi}(\bfx)\equiv\phi_\psi(\bfx+\bfa_x)-\phi_\psi(\bfx)\;({\rm mod}\;2\pi)$. The latter case then has $(-1)^{b_1+1}\sin(\phi^{(\bfa)}_{\psi}(\bfb'))$, while the ancilla is measured in $Y$-basis. Both estimators should be scaled by $\|\breve{\psi}\|_1\sgn(\breve{c}_{\bfa})$ to get the complete estimator of Eq.~\eqref{eq:main_phase_fidelity_estimator}. 

Importantly, let us consider when we estimate the fidelities with $M$ target states $\left\{\ket{\psi_1},\ket{\psi_2},\ldots,\ket{\psi_M}\right\}$ such that all elements share the same phase-stripped state $\ket{\breve{\psi}}$. Then all target states adopt the same $l_{2\alpha}$-sampling for FOFE, the $l_{2\alpha}$-sampling of $\ket{\breve{\psi}}$. It means that we use the same measurement circuit for all target states. Consequently, after getting the measurement outcome $\bfb$ we can use to calculate \emph{many} estimators $(-1)^{b_1}\cos(\phi^{(\bfa)}_{\psi_i}(\bfb))$ or $(-1)^{b_1+1}\sin(\phi^{(\bfa)}_{\psi_i}(\bfb))$ following the phase structure of each $\psi_i\;(i\in [M])$. Therefore, we just need to rescale the failure probability for each phase state to $M^{-1}\delta_f$. So that the total failure probability is $\delta_f$. Finally, using Hoeffding inequality, we prove Thm.~\ref{main:thm_1}.
\begin{figure}[t]
    \includegraphics[width=\columnwidth]{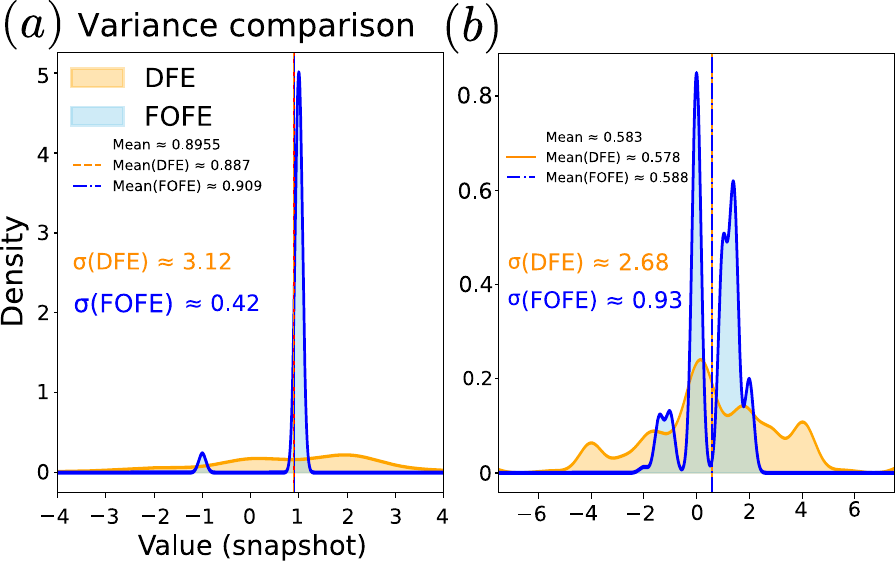}
    \caption{The ($n=7$)-qubit estimation variance of fidelity estimation using FOFE and one qubit ancilla. We compared our result with DFE.~\cite{flammia2011,julia2025}. (a) The target is a $3$rd-order complete hypergraph state $\ket{K_7}$, and the input state is $\ket{K_7}$ with a depolarizing noise. (analytical fidelity~$\simeq0.8955$), and $10000$ sampling copies were~used. (b) The target state is $T^{\otimes n}\ket{K_7}$ with $10000$ copies.}\label{fig:ved_estimation_variance} 
\end{figure}
\begin{figure}[t]
    \includegraphics[width=\columnwidth]{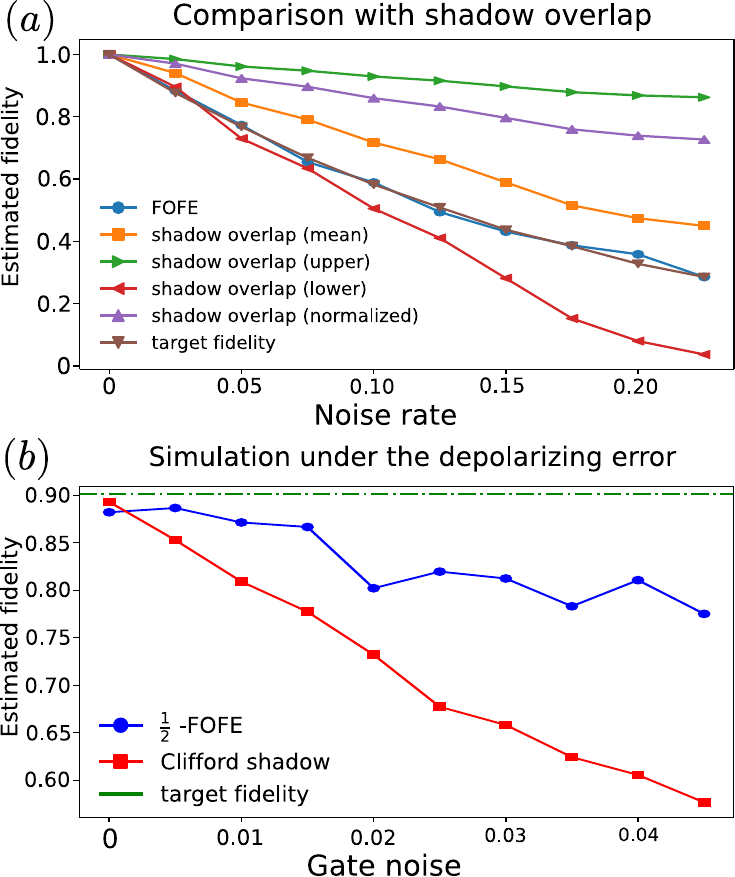}
    \caption{(a) FOFE for the target as $7$-qubit third-order complete hypergraph state $\ket{K_7}$ that undergoes a single-qubit noise $\mathcal{N}$ (uniform operation from $\left\{I,X,Y,Z,H,S\right\}$) with a given rate. We compared our result with shadow overlap estimation~\cite{huang2024_sv} with a single-qubit random Pauli shadow. The (upper) is the right side of Eq.~\eqref{eq:shadow_overlap_bound}, and the (lower) is the left side. The (mean) value is their arithmetic mean. The (normalized) refers to the normalized shadow overlap in Ref.~\cite{huang2024_sv}. The target value was obtained by estimating $\braket{K_7|\mathcal{N}(\ket{K_7}\bra{K_7})|K_7}$ via the direct measurement  with $30000$ copies. (b) Noise robustness comparison between $\frac{1}{2}$-FOFE and random-Clifford shadow~\cite{huang2020}. Here, the target states and input states, with some global depolarizing (white) noise, are the same as $\ket{{\rm Dic}(6,3)}$ followed by the complete-hypergraph CCZ operation. We gave a random Clifford operation using the Qiskit (V$1.4.5$) random\textunderscore clifford().to\textunderscore circuit() method. Each two-qubit gate undergoes the depolarizing noise with a given gate noise rate.}\label{fig:shadowoverlap_and_classical_shadow}
    
\end{figure}

The nonlinear post-processing ($\bfb'\mapsto \phi^{(\bfa)}_{\psi
}(\bfb')$) of the measurement outcome, in the above process, is quite different from linear processing for the original DFE in which we estimate the expectation value of the Pauli operator $P$~\cite{leone2023}. To be specific, $P$ is twirled by single qubit Clifford operations $V\in {\rm Cl}_1^{\otimes n}$ into $Z^{\bfu}$ for some $\bfu\in \bfF_2^n$. Therefore, to estimate $\tr{\rho P}$, the state $\rho$ is twirled by $V^{\dag}$ and measurement is done to obtain $\bfb'\in \bfF_2^n$, which is processed by a linear functional $\bfb'\mapsto \bfu\cdot \bfb'$ (binary inner product) to output $(-1)^{\bfu\cdot \bfb'}$. 

\textit{Numerical results}---Figure~\ref{fig:ved_estimation_variance} shows that FOFE exhibits a drastic improvement in the estimation variance with a fixed number of copies, compared to the conventional DFE~\cite{julia2025}. In Fig.~\ref{fig:ved_estimation_variance} (a), since the target state here is a hypergraph state, the FOFE estimator only outputs $1$ or $-1$. This is because the target state does not have a complex phase, so that ${\rm Im}(D^{(\bfa)})=0$. Figure~\ref{fig:ved_estimation_variance} (b) implies that as the target-state phase becomes diversified, so do the possible FOFE-estimator values. The sampling complexity of FOFE increases as the target state becomes farther from the phase-state manifold. Nevertheless, in SM~\cite{supple}, we showed that even when the target state is drawn Haar-randomly, $\frac{1}{2}$-FOFE still provides a constant-factor improvement.

Next, we compared our results with the shadow overlap estimator~\cite{huang2024_sv} and random Clifford shadows~\cite{huang2020}. To begin with, the shadow overlap can be efficiently done with Pauli measurements. The estimator $\widehat{\omega}$ estimates the expectation value of a proxy operator $L$, for which $\ket{\psi}$ is the eigenstate associated with the largest unit eigenvalue ($L\ket{\psi}=\ket{\psi}$). However, that expectation value $\tr{\rho L}$ generally deviates from the target fidelity, except in certain special noise models such as white noise~\cite{huang2024_sv,sun2025}.    More explicitly, given that the target state is a phase state, shadow overlap offers a fidelity interval as follows, 
\begin{align}\label{eq:shadow_overlap_bound}
    n(\tr{\rho L}-1)+1\le \braket{\psi|\rho|\psi}\le \tr{\rho L}=\mathbb{E}(\widehat{\omega}). 
\end{align}
See Refs.~\cite{huang2024_sv,zhu2019} for details. From Fig.~\ref{fig:shadowoverlap_and_classical_shadow} (a), we observe that both the upper and lower bounds, and their average, deviate from the target fidelity as the noise in the input state increases, whereas FOFE remains close to it. 

Fig.~\ref{fig:shadowoverlap_and_classical_shadow} (b) shows the noise robustness comparison between FOFE and random-Clifford shadow~\cite{huang2020}. We briefly introduce the Clifford-shadow process: we enact uniformly chosen Clifford unitary $U\in {\rm Cl}_n$ to each copy $\rho$, then measure with the computational bases to obtain the outcome $\bfb\in \bfF_2^n$, and then the estimator is $(2^n+1)|\braket{\bfb|U|\psi}|^2-1$.  As we increase the gate noise during the mid-circuit operation, we see that the bias of FOFE manifests a more moderate increase. 

\begin{figure}[t]
    \includegraphics[width=\columnwidth]{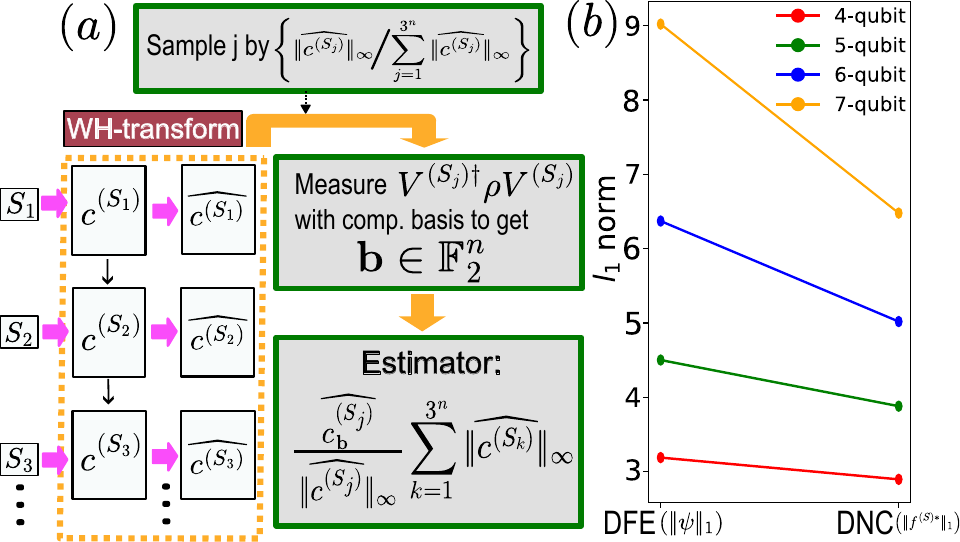}
    \caption{(a) Schematic illustration of DNC-based algorithm for NLDFE. 
    (b) Improvement in $l_1$-norm of $100$-copies of Haar random pure states \emph{via} the DNC-based algorithm.}\label{fig:DnC-based_nldfe_scheme}
    
\end{figure}

\textit{Nonlinear DFE (NLDFE).}---The key idea behind the FOFE is to trade the complex diagonal gate in the Hadamard test circuit for nonlinear classical post-processing of the Pauli measurement outcomes. We can apply a similar technique to the DFE by generalizing the Pauli operators in Eq.~\eqref{eq:Pauli_expansion} to an over-complete set of diagonal operators conjugated by single-qubit Clifford operators. In what follows, we briefly describe the resulting $\alpha$-\emph{nonlinear} DFE ($\alpha$-NLDFE) as an auxiliary method that retains only the Pauli measurements. We fix $\alpha=\frac{1}{2}$  and give a brief overview of $\frac{1}{2}$-NLDFE in this text, leaving the general expression and details in SM~\cite{supple}. 

To begin with the target state $\ket{\psi}$, there exists a function $f:[0,2\pi]^{2^n}\times \left\{I,H,HS\right\}^{\otimes n}\rightarrow \mathbb{C}$ such that 
\begin{align}\label{eq:main_general_NLDFE}
    \ket{\psi}\bra{\psi}=\frac{1}{2^n}\sum_{V\in \left\{I, H, HS\right\}^{\otimes n}}\int_{[0,2\pi]^{2^n}}d\phi f(\phi,V)\;VD(\phi)V^{\dag},
\end{align}
or equivalently $f\mapsto \psi$. Our NLDFE scheme becomes DFE if we restrict the diagonal operators to just Pauli $Z$ operators. 

We define {\small $\|f\|_1\equiv \frac{1}{2^n}\sum_{V\in \left\{I, H, HS\right\}^{\otimes n}}\int_{[0,2\pi]^{2^n}}d\phi|f(\phi,V)|$}. Then, we set the $\frac{1}{2}$-NLDFE scheme: Sample $(\phi,V)$ from the probability distribution {\small$
\left\{\frac{|f(\phi,V)|d\phi}{\|f\|_{1}}\right\}$}, and then estimate $\tr{\rho VD(\phi)V^{\dag}}=\sum_{\bfx\in \bfF_2^n}\braket{\bfx|V^{\dag}\rho V|\bfx}e^{i\phi(\bfx)}$. This can be done by measuring $V^{\dag}\rho V$ with the computational basis and taking the outcome $\bfx\in \bfF_2^n$ with which the estimator is $\cos(\phi(\bfx)+\arg(f))$, where $f=|f|e^{i\arg(f)}$. 
Compared to DFE, we now take the \emph{nonlinear} mapping of the exponent. The square of the $\|f\|_1$ can be shown~\cite{supple} to quantify the sampling overhead of $\frac{1}{2}$-NLDFE. It is desirable to find $f^{\ast}=\mathrm{argmin}_{f\mapsto \psi}\|f\|_1$, which guarantees $\|f^*\|_1\leq\|\psi\|_1$. 

For NLDFE, we have infinitely many coefficients of the overcomplete support $\{VD(\phi)V^{\dag}\}$. Therefore, finding the optimal $f^{\ast}$ is an extremely hard problem. Following that, we propose a divide-and-conquer(DNC)-based sub-optimal algorithm, which is illustrated in Fig.~\ref{fig:DnC-based_nldfe_scheme}. It offers a deterministic solution and still guarantees a tighter sampling bound compared to $\|\psi\|^2_1$. To be specific, we divide Pauli coefficients into the qubit-wise-commuting (QWC) groups~\cite{julia2025} $S_j\subset \mathcal{P}_n/\bfZ_4\;(j\in [3^n])$, such that all elements are conjugated into Pauli $Z$-groups by the same single-qubit Clifford operator. We denote each partitioned-coefficient vector as $c^{(S_j)}$. Then we can find a solution $f$ and following NLDFE algorithm such that $\|f\|_1=\sum_{j\in [3^n]}\|\widehat{c^{(S_j)}}\|_{\infty}\le \|\psi\|_1$, where $\|c\|_\infty\equiv \max_{\bfa\in \bfF_2^{2n}}\left\{|c_{\bfa}|\right\}$ and the caret refers to the Walsh--Hadamard (WH) transform~\cite{scheibler2015,alman2023}, $\widehat{c^{(S)}}_{\bfb}\equiv \sum_{\bfa\in \bfF_2^n}c^{(S)}_{\bfa}(-1)^{\bfa\cdot \bfb}$. See SM~\cite{supple} for details.

\textit{Remarks.}---We introduced a sample-optimal method, $\alpha$-FOFE ($\alpha\in \left\{\frac{1}{2},1\right\}$), for the pure-state fidelity estimation. Here we employed a single-qubit ancilla and a single fan-out gate consisting $n$ CNOTs, enabling the estimation of the fidelity with an arbitrary target state $\ket{\psi}$ with $\mathcal{O}\left(\|\breve{\psi}\|^{1/\alpha}_{2-2\alpha}\right)$ sampling copies, which reduces to $\mathcal{O}(1)$ for the target phase states. The gate efficiency is obtained by replacing the complex diagonal gates with nonlinear post-processing of Pauli measurements. Lastly, we proposed a nonlinear variant of the DFE that achieves sampling reduction while still requiring only Pauli measurements. Beyond the sampling optimization, several challenges remain to be addressed for the FOFE. One is to check whether $n$ CNOTs are the minimal resources for the sample-optimal FE for phase states. Second, although there are many cases where $l_{2\alpha}$-sampling is efficient, the required sampling copies for input state $\rho$ may not be efficient. Hence, exploring examples of achieving efficiency for both procedures may be future work.

\textit{Acknowledgments.}---The authors thank Hyukjoon Kwon, Huangjun Zhu, Jinzhao Sun, and Kento Tsubouchi for their thoughtful discussions. This work was supported by the National Research Foundation of Korea (NRF) grant funded by the Korea government (MSIT) (RS-2023-00237959, RS-2024-00353348, RS-2024-00413957, RS-2024-00438415, and RS-2023-NR076733), the Institute of Information \& Communications Technology Planning \& Evaluation (IITP) grant funded by the Korea government (MSIT) (RS-2020-II201606 and RS-2024-00437191, RS-2025-02219034), the Korea University Grant, and the Institute of Applied Physics at Seoul National University.

\clearpage

\widetext
\begin{center}
	\textbf{\large Supplemental Materials: Sample- and Hardware-Efficient Fidelity Estimation by Stripping Phase-Dominated Magic}
\end{center}

\tableofcontents

\section{Preliminaries: Direct fidelity estimation (DFE)}

Before the technical details of our main results, let us introduce the conventional direct fidelity estimation (DFE)~\cite{flammia2011,leone2023,julia2025,cha2025} formalism as a preliminary. DFE is one of the famous fidelity estimation schemes, especially when the target state is pure. The main objective of DFE is to estimate the fidelity between the given unknown state $\rho$ and the target pure state $\ket{\psi}\bra{\psi}$. We should note that the $\tr{\ket{\psi}\bra{\psi}^2}=\tr{\ket{\psi}\bra{\psi}}=1$, hence if we consider the following Pauli-Liouville decomposition $\ket{\psi}\bra{\psi}=\frac{1}{2^n}\sum_{\bfa\in \bfF_2^{2n}}\braket{\psi|T_\bfa|\psi}T_{\bfa}$, where $T_{\bfa}\equiv \bigotimes_{i=1}^{n}i^{a_{ix}a_{iz}}X^{a_x}Z^{a_z}\;(\bfa\in \bfF^{2n}_2)$, then we note that 
\begin{align}
	\tr{\ket{\psi}\bra{\psi}^2}=\frac{1}{2^n}\sum_{\bfa\in \bfF^{2n}_2}\braket{\psi|T_{\bfa}|\psi}^2=1.
\end{align}
Therefore, we can regard the elements $\left\{\frac{c_\psi(\bfa)^2}{2^n}\bigg|c_\psi(\bfa)\equiv \braket{\psi|T_{\bfa}|\psi}\right\}_{\bfa\in \bfF_2^{2n}}$ as a probability distribution. Next, we rewrite the target fidelity $\braket{\psi|\rho|\psi}$ as,
\begin{align}\label{eq:original_DFE}
	\braket{\psi|\rho|\psi}=\frac{1}{2^n}\sum_{\bfa\in \bfF_2^{2n}}c_\psi(\bfa)\tr{T_{\bfa}\rho}=\frac{1}{2^n}\sum_{\bfa\in \bfF_2^{2n}}c_\psi(\bfa)^2\frac{\tr{\rho T_{\bfa}}}{c_\psi(\bfa)}. 
\end{align}
Hence, the algorithm is as follows. We sample $\bfa\in \bfF^{2n}_2$ from the distribution $\left\{\frac{c_\psi(\bfa)^2}{2^n}\right\}_{\bfa\in \bfF_2^{2n}}$. Then, we take the one-shot measurement~\cite{leone2023} which estimates $\frac{\tr{\rho T_{\bfa}}}{c_\psi(\bfa)}$. Here, we should note that $c_{\psi}(\bfa)\ne 0$ because such a case cannot be sampled. From this algorithm, we get the unbiased estimator of $\braket{\psi|\rho|\psi}$.

The problem is how to make the one-shot measurement for the estimation of  $\frac{\tr{\rho T_{\bfa}}}{c_\psi(\bfa)}$. Here, we introduce a practical method to achieve this by using only Pauli measurements.  We note that $c_{\psi}(\bfa)$ is already known, and that 
\begin{align}
	\tr{\rho T_{\bfa}}=\tr{\rho\left(\frac{I+T_{\bfa}}{2}+(-1)\times \frac{I-T_{\bfa}}{2}\right)}=\tr{\rho\Pi_0}-\tr{\rho \Pi_1},
\end{align}
where $\Pi_{\bfa}^p\equiv \frac{I+(-1)^{p}T_{\bfa}}{2}\;(p\in \bfF_2)$ is the projector. The above equation directly leads to the following estimation algorithm: we measure $\rho$ with the positive operator-valued measurement (POVM) $\left\{\Pi_0,\Pi_1\right\}$. Next, we explain how to measure with this POVM via Pauli measurements. There exists single qubit Clifford operators (i.e., tensor product of single qubit Clifford operator) $V$ such that $T_{\bfa}=VZ^{\bfa'}V^{\dag}$ for some $\bfa'\in \bfF_2^n$. Therefore, $\Pi_{\bfa}^p=V\frac{I+(-1)^pZ^{\bfa'}}{2}V^{\dag}$ and hence we obtain that 
\begin{align}
	\tr{\rho \Pi_{\bfa}^{p}}=\tr{V^{\dag}\rho V\frac{I+(-1)^pZ^{\bfa'}}{2}}=\tr{V^{\dag}\rho V\times \sum_{\bfx\in \bfF_2^n,\bfa'\cdot \bfx=p}\ket{\bfx}\bra{\bfx}}=\sum_{\bfx\in \bfF_2^n,\bfa'\cdot \bfx=p}\braket{\bfx|V^{\dag}\rho V|\bfx}.
\end{align}
Therefore, we measure the POVM by first twirling $\rho$ by a single-qubit Clifford operators $V^{\dag}$ and measure with the computational basis, checking $\bfa'\cdot\bfx$ is $0$ or $1$. The linearity of $\bfa'\cdot \bfx$ is why we call the conventional DFE a linear-DFE.

Furthermore, we can generalize the above conventional DFE scheme based on sampling from the $l_2$ distribution of Pauli coefficients, which we refer to as $l_2$-sampling~\cite{guo2024}, to $l(2\alpha)\;(\alpha\in \mathbb{R}^{+}_{\ge 0})$-sampling.
To do so, we first define an important measure of magic (or non-stabilizerness), the \emph{stabilizer Rényi entropy}.

\begin{definition}~\cite{leone2022}\label{eq:def_SRE}
	Given $0<\alpha<1$ and a pure quantum state $\ket{\psi}$, \emph{$\alpha$-stabilizer Renyi entropy ($\alpha$-SRE)} of $\ket{\psi}$ is defined as,
	\begin{align}
		M_{\alpha}(\psi)\equiv\frac{1}{1-\alpha}\log_2\left(\frac{1}{2^{\alpha n}}\sum_{\bfa\in \bfF^{2n}_2}|\braket{\psi|T_{\bfa}|\psi}|^{2\alpha}\right)-n,
	\end{align}
	and $M_{0}=\lim_{\gamma\rightarrow 0^+}M_{\gamma}$. We also call $\widetilde{M}_{\bfa}(\psi)\equiv\frac{1}{2^n}\sum_{\bfa\in \bfF^{2n}_2}|\braket{\psi|T_{\bfa}|\psi}|^{2\alpha}$ as linear $\alpha$-SRE.
\end{definition}

We also call $\widetilde{M}_{0}$ as $l_0$-norm (or \emph{Pauli rank}~\cite{bu2019}), $ \widetilde{M}_{\frac{1}{2}}$ as \emph{Pauli $l_1$-norm (or stabilizer negativity)}~\cite{chen2024_magic_mps,howard2017}. In particular, the Pauli $l_1$-norm ($l_1$-norm, shortly) is equal to the $l_1$-norm of Pauli coefficients $\left\{\frac{1}{2^n}\braket{\psi|T_\bfa|\psi}\right\}_{\bfa\in \bfF_2^{2n}}$. This norm can be generalized to the arbitrary density matrices as, 
\begin{align}\label{eq:def_stab_norm}
	\widetilde{M}_{\frac{1}{2}}(\sigma)=\sqrt{M_{\frac{1}{2}}(\sigma)}=\frac{1}{2^n}\sum_{\bfa\in \bfF_2^{2n}}|\tr{\sigma T_{\bfa}}|.
\end{align}
Note that each square of Pauli coefficients in $\sigma$ sums to $2^n\tr{\sigma^2}$, $2^n$-factorized purity of $\sigma$. For notational convenience, we denote $l_0$-norm of $\sigma$ as $\|\sigma\|_0$, and the $l_1$-norm as $\|\sigma\|_1$. 

Next, we recall the fidelity expression Eq.~\eqref{eq:original_DFE}, but change the sampling distribution into $\left\{\frac{|c_{\psi}(\bfa)|^{2\alpha}}{\sum_{\bfa\in \bfF_2^{2n}}|c_{\psi}(\bfa)|^{2\alpha}}\right\}_{\bfa\in \bfF_2^{2n}}$ that leads to,  

\begin{align}\label{eq:generalized_alpha_DFE}
	\braket{\psi|\rho|\psi}&=\sum_{\bfa\in\bfF_2^n }c_{\psi}(\bfa)\tr{\rho T_{\bfa}}\nonumber\\&=\sum_{\bfa\in \bfF^{2n}_2}|c_{\psi}(\bfa)|^{2\alpha}\left\{\tr{\rho T_{\bfa}}|c_{\psi}(\bfa)|^{-2\alpha+1}\sgn(c_{\psi}(\bfa))\right\}\nonumber\\&=2^{(1-\alpha)(M_{\alpha}(\psi)+n)-\alpha n}\sum_{\bfa\in \bfF^{2n}_2}\frac{|c_{\psi}(\bfa)|^{2\alpha}}{2^{(1-\alpha)(M_{\alpha}(\psi)+n)-\alpha n}}\left\{\tr{\rho T_{\bfa}}|c_{\psi}(\bfa)|^{-2\alpha+1}\sgn(c_{\psi}(\bfa))\right\}\nonumber\\&=2^{(1-\alpha)(M_{\alpha}(\psi)+n)-\alpha n}\sum_{\bfa\in \bfF^{2n}_2}\frac{|c_{\psi}(\bfa)|^{2\alpha}}{2^{(1-\alpha)(M_{\alpha}(\psi)+n)-\alpha n}}\left\{\tr{\rho (\Pi_{\bfa}^0-\Pi_{\bfa}^1)}|c_{\psi}(\bfa)|^{-2\alpha+1}\sgn(c_{\psi}(\bfa))\right\}
	\nonumber\\&=2^{(1-\alpha)(M_{\alpha}(\psi)+n)-\alpha n}\sum_{\bfa\in \bfF^{2n}_2}\frac{|c_{\psi}(\bfa)|^{2\alpha}}{2^{(1-\alpha)(M_{\alpha}(\psi)+n)-\alpha n}}\sum_{p\in \bfF_2}\tr{\rho \Pi_{\bfa}^p}(-1)^p\left\{|c_{\psi}(\bfa)|^{-2\alpha+1}\sgn(c_{\psi}(\bfa))\right\},
\end{align}
The result enables us to estimate $\braket{\psi|\rho|\psi}$ following the general scheme as below, with fixed $N,K\in \mathbb{N}$,

\begin{enumerate}
	\item We sample $\bfa$ from $\frac{|c_{\psi}(\bfa)|^{2\alpha}}{\sum_{\bfa\in \bfF^{2n}_2}|c_{\psi}(\bfa)|^{2\alpha}}$.
	\item Measure $\rho$ with the POVM $\left\{\Pi_{\bfa}^0,\Pi_{\bfa}^1\right\}$, where $\Pi_{\bfa}^p\equiv \frac{1+(-1)^{p}T_{\bfa}}{2}\;(p\in \bfF_2)$ to get the outcome $p\in \bfF_2$.
	\item Take the estimator, $m\equiv (-1)^p\left(\sum_{\bfa}|c_{\psi}(\bfa)|^{2\alpha}\right)|c_{\psi}(\bfa)|^{-2\alpha+1}{\rm sgn}(c_{\psi}(\bfa))$.
	\item Repeat above steps $N$ times and get the estimators $m_1,m_2,\ldots,m_{N}$. Then the final estimated value is $\frac{1}{N}\sum_{i=1}^{N}m_i$.
	
	\item Repeat step~$4$, $K$ times to obtain $\hat{m}_1,\hat{m}_2,\ldots,\hat{m}_K$, then the final estimation becomes 
	\begin{align}
		\hat{m}={\rm median}\left\{\hat{m}_1,\hat{m}_2,\ldots,\hat{m}_K\right\}.
	\end{align}
\end{enumerate}

We call such a scheme an $\alpha$-DFE. Hence, the conventional DFE scheme reduces to $1$-DFE ($\alpha\rightarrow 1^{-}$).

\section{Preliminaries: Estimation variance and algorithmic efficiency of $\frac{1}{2}$-DFE}

In the previous section, we introduced the practical estimation routine for the $\alpha$-DFE. Using the form of Eq.~\eqref{eq:generalized_alpha_DFE}, we can also calculate its estimation variance, which quantifies the required sampling copies for a desired accuracy. Moreover, we see that the original DFE index~\cite{flammia2011}, $\alpha\rightarrow 1^{-}$ is not the sample-optimal choice. To see this, the estimation variance is bounded by, 
\begin{align}\label{eq:MSE_upper}
	{\rm Var}(\rho,\psi,\alpha)+\braket{\psi|\rho|\psi}^2= \mathbb{E}\left(\widehat{ \braket{\psi|\rho|\psi}}^2\right)&=\sum_{\bfa\in \bfF^n_2}|c_{\psi}(\bfa)|^{2\alpha}2^{(1-\alpha)(M_{\alpha}(\psi)+n)-\alpha n}\sum_{p\in \bfF^n_2}\tr{\rho \Pi_{\bfa}^p}|c_{\psi}(\bfa)|^{-4\alpha+2}\nonumber\\&=\sum_{\bfa\in \bfF^n_2}|c_{\psi}(\bfa)|^{-2\alpha+2}2^{(1-\alpha)(M_{\alpha}(\psi)+n)-\alpha n}\nonumber\\&=2^{\left[\left\{\alpha (M_{1-\alpha}(\psi)+n)-(1-\alpha)n\right\}+\left\{(1-\alpha)(M_{\alpha}(\psi)+n)-\alpha n\right\}\right]}\nonumber\\&=2^{\alpha M_{1-\alpha}(\psi)+(1-\alpha)M_{\alpha}(\psi)}.
\end{align}
It means that $\log_2\left(\mathbb{E}\left(\widehat{ \braket{\psi|\rho|\psi}}^2\right)\right)$ is a convex combination of two measures, $M_{\alpha}$ and $M_{1-\alpha}$. From now on, we will denote ${\rm Var}(\psi,\alpha)=\max_{\rho}\left\{{\rm Var}(\rho,\psi,\alpha)\right\}$. Since $\braket{\psi|\rho|\psi}^2\le 1$, we shall ignore this term for the scaling of variance.In addition, there is a well-known convexity theorem called \emph{log-sum-exp rule} that is, 
\begin{lemma}\label{lem:log-sum-exp rule}
    For a fixed probaiblity distribution $\left\{p_{\bfa}\right\}_{
    \bfa\in \bfF_2^{2n} 
    }$, both $f(\alpha)=\log_2\left(\sum_{\bfa\in \bfF_2^{2n}}p^{\alpha}_{\bfa}\right)$ and $f(1-\alpha)$ is convex function.
\end{lemma}
\begin{proof}
    $f(\alpha)=\frac{1}{\ln 2}\ln(\bar{f}(\alpha))$, where $\bar{f}(\alpha)=\ln\left(\sum_{\bfa\in \bfF_2^{2n}}p^{\alpha}_{\bfa}\right)$. Hence, proving that $\bar{f}$ is convex is sufficient. By simple calculation, the double derivative, 
    \begin{align}
        \frac{\partial^2 \bar{f}}{\partial\alpha^2}=\sum_{\bfa\in \bfF_2^{2n}}\frac{p^{\alpha}_\bfa}{\sum_{\bfb\in \bfF_2^{2n}}p^\alpha_{\bfb}}\ln^2(p_{\bfa})-\left(\sum_{\bfa\in \bfF_2^{2n}}\frac{p^{\alpha}_\bfa}{\sum_{\bfb\in \bfF_2^{2n}}p^\alpha_{\bfb}}\ln(p_{\bfa})\right)^2,
    \end{align}
    is an estimation variance of the estimator $\ln(p_{\bfa})$ followed by the distribution $\left\{\frac{p^{\alpha}_\bfa}{\sum_{\bfb\in \bfF_2^{2n}}p^\alpha_{\bfb}}\right\}_{\bfa\in \bfF_2^{2n}}$, and is non-negative. Hence $f(\alpha)$ is convex. Convexity of $f(1-\alpha)$ naturally follows from $0\le\frac{\partial^2 \bar{f}(1-\alpha)}{\partial(1-\alpha)^2}=-\frac{\partial^2 \bar{f}(1-\alpha)}{\partial\alpha\partial(1-\alpha)}=\frac{\partial^2 \bar{f}(1-\alpha)}{\partial\alpha^2}$.
\end{proof}

It leads to the following result, 

\begin{corollary}\label{Cor:renyi_convex}
	$\alpha M_{1-\alpha}(\psi)+(1-\alpha)M_{\alpha}(\psi)$ is minimum at $\alpha=\frac{1}{2}$.
\end{corollary}
\begin{proof}
   We remember that $\left\{\frac{1}{2^n}|\braket{\psi|T_{\bfa}|\psi}|^2\right\}_{\bfa\in \bfF_2^{2n}}$ acts as a probability distribution. By Def.~\ref{eq:def_SRE} and Lem.~\ref{lem:log-sum-exp rule}, we note that both $\alpha M_{1-\alpha}$ and $(1-\alpha)M_{\alpha}$ are convex.  Hence its sum, $\alpha M_{1-\alpha}(\psi)+(1-\alpha)M_{\alpha}(\psi)$ is again convex with $\alpha$. Furthermore, take the derivative with $\alpha$ then we see that it hits zero at $\alpha=\frac{1}{2}$. 
\end{proof}

This implies that we have the best estimation accuracy if we could sample the $\bfa$ following $\alpha=\frac{1}{2}$, which matches with the tight sampling scaling for the black-box estimator~\cite{fawzi2024}. In other words, the probability $p(\bfa)$ to get $\bfa\in \bfF^{2n}_2$  is, 
\begin{align}
	p(\bfa)=2^{-\frac{M_{1/2}(\psi)}{2}}|c_{\psi}(\bfa)|.
\end{align}

We call such routine as \emph{$l_1$-sampling} in which ${\rm Var}\left(\psi,\frac{1}{2}\right)=\|\psi\|_1^2$. Again, the original direct fidelity estimation~\cite{flammia2011} was of \emph{$l_2$-sampling} ($\alpha\rightarrow 1^{-}$). The general notation \emph{$l_{2\alpha}$-sampling} naturally follows. In $l_1$-sampling, stabilizer negativity quantifies the required~\cite{howard2017} sampling-copy complexity. In the $\alpha\rightarrow 1^{-}$ case, we recover the well-known result of Ref.~\cite{flammia2011},

\begin{proposition}
	\begin{align}
		\rm{Var}(\psi,1)\le \frac{1}{\gamma(\psi)^2},
	\end{align} where
	\begin{align}
		\gamma(\psi)^2=\min_{\bfa\in \bfF^{2n}_2}\left\{\braket{\psi|T_{\bfa}|\psi}|\braket{\psi|T_{\bfa}|\psi}\ne 0\right\}.
	\end{align}
\end{proposition}

\begin{proof}
	When $\alpha\rightarrow 1^{-}$, $ \rm{Var}(\psi,1)$ is bounded by, 
	
	\begin{align}
		\mathbb{E}\left(\widehat{ \braket{\psi|\rho|\psi}}^2\right)\le 2^{M_{0}(\psi)+\log_2\left(\frac{1}{2^n}\sum_{\bfa\in \bfF^{2n}_2}\braket{\psi|T_{\bfa}|\psi}^2\right)}=2^{M_0(\psi)},
	\end{align}
	where the last equality holds since the purity $\frac{1}{2^n}\sum_{\bfa\in \bfF^{2n}_2}\braket{\psi|T_{\bfa}|\psi}^2=1$.
	Here, 
	
	\begin{align}
		M_{0}=\log_2\left(\left|\left\{\braket{\psi|T_{\bfa}|\psi}|\braket{\psi|T_{\bfa}|\psi}\ne 0,\bfa\in \bfF^{2n}_2\right\}\right|\right)-n.
	\end{align}
	We set $A\equiv 2^{M_0(\psi)+n}$ for convenience. We note that 
	\begin{align}
		2^n=\sum_{\bfa\in \bfF^{2n}_2}\braket{\psi|T_{\bfa}|\psi}^2\ge A\gamma(\psi)^2.
	\end{align}
	Then, we conclude $2^{M_0+n}\le \frac{2^n}{\gamma(\psi)^2}$, and hence $ \rm{Var}(\psi,1)\le \frac{1}{\gamma(\psi)^2}$.
\end{proof}

$\alpha$-DFE requires an efficient $l_{2\alpha}$-sampling of the phase point $\bfa\in \bfF_2^{2n}$. Unfortunately, not every case of the target state satisfies the efficient $l_{2\alpha}$-sampling. Next, we demonstrate that $\frac{1}{2}$-DFE is efficiently simulated for Dicke states~\cite{Brtschi2019}, which is typically observed as efficient cases of $1$-DFE~\cite{flammia2011,cha2025,julia2025}.
\begin{proposition}\label{prop:dicke_efficiency}
	(i) Consider Dicke-$(n,k\le\lfloor\frac{n}{2}\rfloor)$ state that is, $\ket{{\rm Dic}(n,k)}\equiv\frac{1}{\sqrt{\comb{n}{k}}}\sum_{\bfx\in \bfF_2^n,|\bfx|=k}\ket{\bfx}$. Then $l_1$-sampling for Dicke states takes $\mathcal{O}(k^4n)$-time.\\
	(ii) $\|{\rm Dic}(n,k)\|^2_1\le 2^{M_0({\rm Dic}(n,k))}\le\mathcal{O}(n^{2k})$.
\end{proposition}
\begin{proof}
	We first prove (i). Note that
	\begin{align}
		|c_{\psi}(\bfa)|=\frac{1}{2^n\comb{n}{k}}\left|\sum_{|\bfy|=k}(-1)^{\bfa_z\cdot \bfy}\left(\sum_{|\bfx|=k}\delta_{\bfx,\bfa_x+\bfy}\right)\right|.
	\end{align}
	Non-zero terms only occur when $|\bfa_x+\bfy|=k-a+|\bfa_x|-a=k$, where $a=a(\bfa_x,\bfy)$ denotes the number of overlapped $1$'s between $\bfa_x$ and $\bfy$. Hence $|\bfa_x|=2a\le 2k$ should be even. Furthermore, we note that $c_{\psi}(\bfa)=c_{\psi}(\sigma\oplus \sigma(\bfa))$ for any permutation $\sigma\in S_{2^n}$. Therefore, we have $\sum_{\bfa_z}|c_{\psi}(\sigma \bfa_x,\bfa_z)|=\sum_{\bfa_z}|c_{\psi}(\sigma \bfa_x,\sigma \bfa_z)|=\sum_{\bfa_z}|c_{\psi}(\bfa)|$, and then the x-marginal sampling probability depends only on the even Hamming weight of $\bfa_x$. Exactly, we obtain that the marginal probability of $\bfa_x$'s with a given Hamming weight $p$ (even) is,
	\begin{align}
		&\sum_{\bfa_z
			,|\bfa_x|=p}|c_{\psi}(\bfa)|\nonumber\\&=\frac{\comb{n}{p}}{2^n\comb{n}{k}}\sum_{\bfa_z}\left|\sum_{|\bfy|=k,a(\mathbf{1}_p,\bfy)=\frac{p}{2}}(-1)^{\bfy\cdot \bfa_z}\right|\nonumber\\&=\frac{\comb{n}{p}}{2^n\comb{n}{k}}\sum_{q_1=0}^{p}\sum_{q_2=0}^{n-p}\comb{p}{q_1}\cdot \comb{n-p}{q_2}\left|\sum_{|\bfy|=k,a(\mathbf{1}_p,\bfy)=\frac{p}{2}}(-1)^{(y_1+\cdots+y_{q_1})+(y_{p+1}+\cdots+y_{p+q_2})}\right|\nonumber\\&=\frac{\comb{n}{p}}{2^n\comb{n}{k}}\sum_{q_1=0}^{p}\sum_{q_2=0}^{n-p}\comb{p}{q_1}\cdot \comb{n-p}{q_2}\left|\sum_{\bfy'\in \bfF^p_{2},|\bfy'|=\frac{p}{2}}(-1)^{(y'_1+\cdots+y'_{q_1})}\cdot\sum_{\bfy''\in \bfF^{n-p}_2, |\bfy''|=k-\frac{p}{2}}(-1)^{(y''_{p+1}+\cdots+y''_{p+q_2})}\right|\;(\because\;\frac{p}{2}\le k)
		\nonumber\\&=\frac{\comb{n}{p}}{2^n\comb{n}{k}}\sum_{q_1=0}^{p}\sum_{q_2=0}^{n-p}\left\{\comb{p}{q_1}\cdot \comb{n-p}{q_2}\left|\sum_{l=0}^{\min\left\{q_1,\frac{p}{2}\right\}}(-1)^l\comb{q_1}{l}\cdot\comb{p-q_1}{\frac{p}{2}-l}\right|\cdot\left|\sum_{l=0}^{\min\left\{q_2,k-\frac{p}{2}\right\}}(-1)^l\comb{q_2}{l}\cdot\comb{n-p-q_2}{k-\frac{p}{2}-l}\right|\right\}.
	\end{align}
	Here, $\mathbf{1}_p\equiv (1,\ldots,1,0,\ldots,0)$ with $p$ numbers of $1$'s and we set $\comb{a}{b}=\frac{a!}{b!(a-b)!}=0$ if $a<b$.
	Calculating the above equation with all $p\le 2k$ takes $\mathcal{O}(k^4n)$ time. Then we can sample the $\bfa_x$ in the following manner. We first sample the weight $p$ from the distribution $\eta(p)$ expressed as 
	\begin{align}
		\eta(p)=\frac{        \sum_{\substack{\bfa_z
					\\ |\bfa_x|=p
					\;{\rm is\;even}}}|c_{\psi}|}{\sum_{p}        \sum_{\substack{\bfa_z
					\\ |\bfa_x|=p
					\;{\rm is\;even}}}|c_{\psi}|}.
	\end{align}
	Then we uniformly sample $\bfa_x$ among the $n$-binary strings of the same weight $p$. After $\bfa_x$ is chosen, we can sample $\bfa_z$ similarly. We decompose $\bfa_z=\bfa_z^{(1)}\oplus\bfa_z^{(2)}$, where $\bfa_{z}^{(1)}$ is on the positions having $1$ in $\bfa_x$. We note that for an arbitrary permutation $\sigma_1$, $\sigma_2$,
	\begin{align}
		|c_{\psi}(\bfa_x, \sigma_1\bfa_z^{(1)}\oplus \sigma_2\bfa_z^{(2)} )|&=\frac{1}{2^n\comb{n}{k}}\left|\sum_{|\bfy|=k, a(\bfa_x,\bfy)=\frac{p}{2}}(-1)^{\bfy\cdot (\sigma_1\bfa_z^{(1)}\oplus \sigma_2\bfa_z^{(2)})}\right|\nonumber\\&=\frac{1}{2^n\comb{n}{k}}\left|\sum_{\substack{\bfy=\bfy^{(1)}\oplus \bfy^{(2)}\\|\bfy^{(1)}|=\frac{p}{2},|\bfy^{(2)}|=k-\frac{p}{2}}}(-1)^{\sigma_1^{\top}\bfy^{(1)}\cdot \bfa_z^{(1)}+\sigma_2^{\top}\bfy^{(2)}\cdot \bfa_z^{(2)}}\right|
		\nonumber\\&=\frac{1}{2^n\comb{n}{k}}\left|\sum_{\substack{\bfy=\bfy^{(1)}\oplus \bfy^{(2)}\\|\bfy^{(1)}|=\frac{p}{2},|\bfy^{(2)}|=k-\frac{p}{2}}}(-1)^{\bfy^{(1)}\cdot \bfa_z^{(1)}+\bfy^{(2)}\cdot \bfa_z^{(2)}}\right|\nonumber\\&=|c_{\psi}(\bfa_x,\bfa_z )|.
	\end{align}
	
	Therefore, the probability to sample $\bfa_z$ (conditioned on $\bfa_x$) only depends on the Hamming weights of $\bfa^{(1)}_z$ and $\bfa^{(2)}_z$. Therefore, we conclude that sampling $\bfa_z$ also takes $\mathcal{O}(k^3 n)$-time. (i) is proved. 
	
	We move on to the statement (ii). The first inequality is definite by Def.~\ref{eq:def_SRE} and Cor.~\ref{Cor:renyi_convex}. Hence let us prove the next inequality.  We need to count the nontrivial Pauli supports of $\ket{{\rm Dic}(n,k)}$. Note that $\braket{{\rm Dic}(n,k)|T_{\bfa}|{\rm Dic}(n,k)}=\frac{1}{\comb{n}{k}}\sum_{|\bfx|,\bfy|=k}\braket{\bfx|T_{\bfa}|\bfy}=\frac{1}{\comb{n}{k}}\sum_{|\bfx|,\bfy|=k}\delta_{\bfx+\bfa_x,\bfy}(-1)^{(\bfx+\bfa_x)\cdot \bfy}$. If this value is nonzero, similarly to the proof of (i), $|\bfa_x|$ should be even and lower or equal than $2k$. Therefore, total number of nontrivial Pauli support is at most $2^n\times\left(\comb{n}{0}+\comb{n}{2}+\,\ldots,+\comb{n}{2k}\right)=\mathcal{O}(2^n n^{2k})$, where $2^n$ factor came from the freedom of choosing $Z$-part in $T_{\bfa}$. It leads to $2^{M_0({\rm Dic}(n,k))}\le \mathcal{O}(n^{2k})$. 
\end{proof}

\section{Pauli $l_1$-norm of hypergraph states}

One of the famous examples of magic states is the \emph{hypergraph state}~\cite{rossi2013,zhu2019}, a specific phase state~\cite{huang2024_sv,lee2025_shallow}. In this section, we show that the random hypergraph state is a representative example of the impossibility of executing sample-efficient DFE. To do so, we find that its $l_1$-norm is  $\simeq \Theta (2^{0.5n})$ so that the sampling complexity of $\alpha$-DFE is $\mathcal{O}(2^n)$. Before introducing this, we define several basic notations. A \emph{$k$th-ordered controlled-Z gate} is defined as the unitary operation with respect to the computational basis $\ket{\bfx} (\bfx\in \bfF^n_2)$, 
\begin{align}
	C_{\left\{i_1,i_2,\ldots,i_k\right\}}Z\ket{\bfx}=(-1)^{x_{i_1}x_{i_2}\ldots x_{i_k}}\ket{\bfx},
\end{align}
where $i_j\in [n]$, $j\in [k]$. We say it as \emph{multiply controlled Z gate} if the specification of $k$ is unnecessary. If $k=1 (2\;{\rm resp.})$, this is $Z$($CZ$)-gate. If $k=3$, we just say it as \emph{controlled-controlled Z(CCZ)-gate} and denote it $CCZ_{\left\{i_1,i_2,i_3\right\}}$. We now consider a \emph{hypergraph} $G(V,E)$ where $V=[n]$ and $E$ is a set of subsets in $V$ with a maximal size $k\ge 2$. The \emph{$k$th-ordered hypergraph state} $\ket{G(V,E)}$ is defined as~\cite{rossi2013}, 
\begin{align}\label{main:def_hypergraph_state}
	\ket{G(V,E)}\equiv \left(\prod_{A\in E}C_A Z\right)\ket{+}^{\otimes n}=\frac{1}{\sqrt{2^n}}\sum_{\bfx\in \bfF^n_2}\left(\prod_{A\in E}C_A Z\right)\ket{\bfx}=\frac{1}{\sqrt{2^n}}\sum_{\bfx\in \bfF^n_2}(-1)^{P_{G}(\bfx)}\ket{\bfx}.
\end{align}
Here, $P_{G}$ denotes the corresponding $k$th-degree Boolean polynomial. We can see that the $U_E\equiv\prod_{A\in E}C_A Z$ is of $k$th-ordered Clifford hierarchy~\cite{pllaha2020,anderson2024}. If all and only $k$th-order multiple controlled Z gates are filled, we call $\ket{G(V,E)}$ a $k$th-order complete hypergraph state, which is also denoted as $\ket{K_k}$. 

We consider when the order of the target hypergraph state $\ket{\psi}$ has the order $3$ (CCZ connections). We remember that these states are uniquely described as specific third-degree Boolean functions. Suppose we are given a Boolean function $f:\bfF^n_2\rightarrow \bfF_2$. We define the directional derivative of $f$ as $D_{\bfv}f(\bfa)\equiv f(\bfa+\bfv)+f(\bfa)$ We first define the linear structure~\cite{kacsikci2016} of $f$ as follows. 
\begin{definition}
	Given a Boolean function $f$, the \emph{linear structure} of $f$ is defined as,
	\begin{align}
		{\rm LS} (f)\equiv \left\{\bfv\in \bfF^n_2|D_\bfv f\;{\rm is\;constant}\right\}.  
	\end{align}
\end{definition}

\begin{corollary}~\cite{kacsikci2016}
	${\rm LS}(f)$ is linear subspace of $\bfF^n_2$.
\end{corollary}
\begin{proof}
	Suppose $\bfu,\bfv\in \linearstructure (f).$ Then $f(\bfu+\bfa)+f(\bfa)+f(\bfv+\bfa)+f(\bfa)=f(\bfa+\bfu)+f(\bfa+\bfu+\bfu+\bfv)$ is constant for $\bfa\in \bfF^n_2$. This is still constant if we translate $\bfa$ to $\bfa+\bfu$ ($\because\;\forall \bfb\in \bfF^n_2,\; \bfb+\bfb=\mathbf{0}$). Therefore, we conclude $f(\bfa+\bfu+\bfv)+f(\bfa)$ is constant and then $\bfu+\bfv\in \linearstructure(f)$. 
\end{proof}

Next, we recall the Walsh-Hadamard transform~\cite{scheibler2015}, which is the correspondence of the Fourier transform in $GF(2^n)\simeq \bfF^n_2$~\cite{macwilliams1977},

\begin{align}
	\hat{f}(\bfu)\equiv\frac{1}{2^n}\sum_{\bfa\in \bfF^n_2}(-1)^{\bfu\cdot \bfa+f(\bfa)} .
\end{align}

Now, we introduce a known result. 
\begin{proposition}~\cite{charpin2005,kacsikci2016,bravyi2019}
	Let the Boolean function $f$ be quadratic. Then for all $\bfu\in \bfF^n_2$, we have 
	\begin{align}
		|\hat{f}(\bfu)|\in\left\{0,2^{\frac{{\rm dim}(\linearstructure(f))-n}{2}}\right\}.
	\end{align}
	Here, $\hat{f}(\bfu)$ can be calculated in $\mathcal{O}(n^3)$ time and memory. 
	Furthermore, $|{\rm supp}(\hat{f})|=2^{n-{\rm dim}(\linearstructure(f))}$. In other words, $\sum_{\bfu\in \bfF^n_2}|\hat{f}(\bfu)|=2^{\frac{n-{\rm dim}(\linearstructure(f))}{2}}$.
\end{proposition}

Now, we consider a third-ordered hypergrpah state $\ket{\psi}=\ket{G(V,E)}$. Then its Pauli $l_1$ norm is 
\begin{align}
	\|\psi\|_1=\widetilde{M}_{\frac{1}{2}}(\psi)&=\frac{1}{4^n}\sum_{\bfx\in \bfF^n_2}\sum_{\bfy\in \bfF^n_2}\left|\sum_{\bfa\in \bfF^n_2}(-1)^{P_G(\bfx+\bfa)+P_{G}(\bfa)+\bfa\cdot \bfy}\right|=\frac{1}{2^n}\sum_{\bfx\in \bfF^n_2}\sum_{\bfy\in \bfF^n_2}\left|\widehat{D_{\bfx}P_G}(\bfy)\right|\nonumber\\&=\frac{1}{2^n}\sum_{\bfx\in \bfF^n_2}2^{\frac{n-{\rm dim}\left(\linearstructure\left(\widehat{D_{\bfx}P_G}\right)\right)}{2}}=\mathbb{E}_{\bfx\in \bfF^n_2}\left\{2^{\frac{n-{\rm dim}\left(\linearstructure\left(\widehat{D_{\bfx}P_G}\right)\right)}{2}}\right\}.
\end{align}
Moreover, by Eq.~\eqref{eq:def_SRE}, we obtain that $\widetilde{M}_{0}(\psi)=\mathbb{E}_{\bfx\in \bfF^n_2}\left\{2^{n-{\rm dim}\left(\linearstructure\left(\widehat{D_{\bfx}P_G}\right)\right)}\right\}.$

\begin{figure}[t]
	\includegraphics[width=6.5cm]{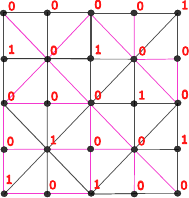}
	\caption{Graph representation (pink edges) corresponding to the adjacency matrix $N_{G}(\bfx)$ from the $25$ qubit Union Jack state~\cite{zhu2019}. Here, $\bfx$ is chosen following each qubit's $\bfF_2$ value (number next to each vertex). }\label{fig:Union_Jack_directional_derivative}
	
\end{figure}

Given that $P_G(\bfa)=\sum_{i,j,k\in [n]}c_{ijk}a_ia_ja_k$ is of third-degree, $D_{\bfx} P_G$ is of second-degree. Therefore we notice that the $\linearstructure (D_\bfx P_G)$ is a null space of some matrix defined by $D_{\bfx}P_G$, say $N_G(\bfx)$. Explicitly, we conclude that by the rank-nullity theorem and the convexity of exponential functions ($\because$ base is larger than $1$),
\begin{align}\label{eq:linear_entropy_hypergrpah_bound}
	\widetilde{M}_{\frac{1}{2}}(\psi)=\mathbb{E}_{\bfx\in \bfF^n_2}\left\{2^{\frac{\rank(N_G(\bfx))}{2}}\right\}\ge 2^{\frac{1}{2}\mathbb{E}_{\bfx\in \bfF^n_2}\left\{\rank (N_G(\bfx))\right\}},
\end{align}
where $\rank$ indicates the rank over $\bfF_2$ (i.e., binary rank) and the $n$ by $n$ matrix $N_G(\bfx)$ is defined by,
\begin{align}
	N_G(\bfx)_{m,k}=\sum_{(i,j),c_{(i,j,m)}=1}\left(\delta_{k,i}x_j+\delta_{k,j}x_i\right)=\sum_{i,c_{(i,m,k)}=1}x_i.
\end{align}

Here, $(i,j,k)$ does not differ by the translation of indices. We note that its collections for all $\bfx\in \bfF^n_2$ are subspace of the space of hollow-symmetric matrices. If we regard the matrix $N_G(\bfx)$ as the adjacency matrix of some graph, we can obtain a graph representation of the $N_G(\bfx)$. As we see Fig.~\ref{fig:Union_Jack_directional_derivative}, an edge on each side of triangular faces is determined by the $\bfF_2$ value of opposite vertices (qubits). If the opposite vertices have an even number of $1$'s, then the edge vanishes. Following that, we can interpret the expectation over uniform binary string $\bfx\in \bfF^n_2$ as the expectation over uniform random graphs whose edges only reside on the sides of CCZ-connections.
Therefore, the lower and upper bounds of the estimation variance of $1$-DFE are written by (ignoring $\braket{\psi|\rho|\psi}$), 
\begin{align}\label{eq:var_graph_bound}
	2^{\mathbb{E}_{\bfx\in \bfF^n_2}\left\{\rank (N_G(\bfx))\right\}}\le{\rm Var}\left(\rho,\psi,\frac{1}{2}\right)\le\mathbb{E}_{\bfx\in \bfF^n_2}2^{\rank(N_G(\bfx))},
\end{align}
where the convexity of the square function is used for an upper bound.

From Eq.~\eqref{eq:linear_entropy_hypergrpah_bound}, we can estimate the $\widetilde{M_{\frac{1}{2}}}(\psi)$. To be specific, we can sample $M$ copies of $\bfx$'s uniformly and independently (say $\bfx_1,\bfx_2,\ldots,\bfx_M$), then we take the $\alpha=\frac{1}{n}\log_2\left(\frac{1}{M}\sum_{i=1}^{M}2^{\frac{\rank(N_G(\bfx_i))}{2}}\right)$ which leads to the final estimation $2^{\alpha n}$ after fitting with various $n$'s. We note that each $\rank$ in the summation can be efficiently calculated in $\mathcal{O}(n^3)$ time. 

Moreover, we can exactly calculate the $\frac{1}{2}$-SRE of complete hypergraph states. It is previously considered in Ref.~\cite{chen2024}, and we present a different approach based on random graph theory. Let us give a detail. 
Let $\mathcal{E}^{(n)}_c$~\cite{chen2024} be the ensemble of random $c$-uniform hypergraph with $n$-vertices (qubits). We also define $\mathbb{E}_{G\sim \mathcal{E}^{(n)}_c}$ as average value over randomly chosen graphs $G$ from $\mathcal{E}^{(n)}_c$. Next, we need a lemma as follows: 
\begin{lemma}\label{lem:uniform_choosing_2_graph}
	Suppose that $G$ is a complete $3$-hypergraph. Given a real-valued function $f$ having $N_G(\bfx)$ as the argument,
	\begin{align}
		\mathbb{E}_{\bfx\in \bfF^n_2}\left\{f(N_G(\bfx))\right\}=\mathbb{E}_{G'\in \mathcal{E}^{(n)}_2}\left\{f(N_{G'})\right\},
	\end{align}
	where the $N_{G'}$ denotes the adjacency matrix of the graph $G'$.
\end{lemma}
\begin{proof}
	We remember that $N_G(\bfx)$ is a $2$-graph. Consider the probability of $N_G(\bfx)$ that an edge occurs on the $(i,j)$-th vertices. The opposite vertices, which are endowed uniform $(0,1)$-values, determine the edge occurrence. In other words, if the even (odd resp.) numbers of opposite vertices get the $1\;(0)$, then the edge is (not) formed. Therefore, no matter how many of the opposite vertices there are, the corresponding edge occurrence has a half probability. Finally, since we know that all edge occurs independently with the same probability ($\frac{1}{2}$-Erdos-Renyi graph~\cite{erdos1963}), we conclude that edge occurs uniformly.  
\end{proof}

\begin{figure}[t]
	\includegraphics[width=\columnwidth]{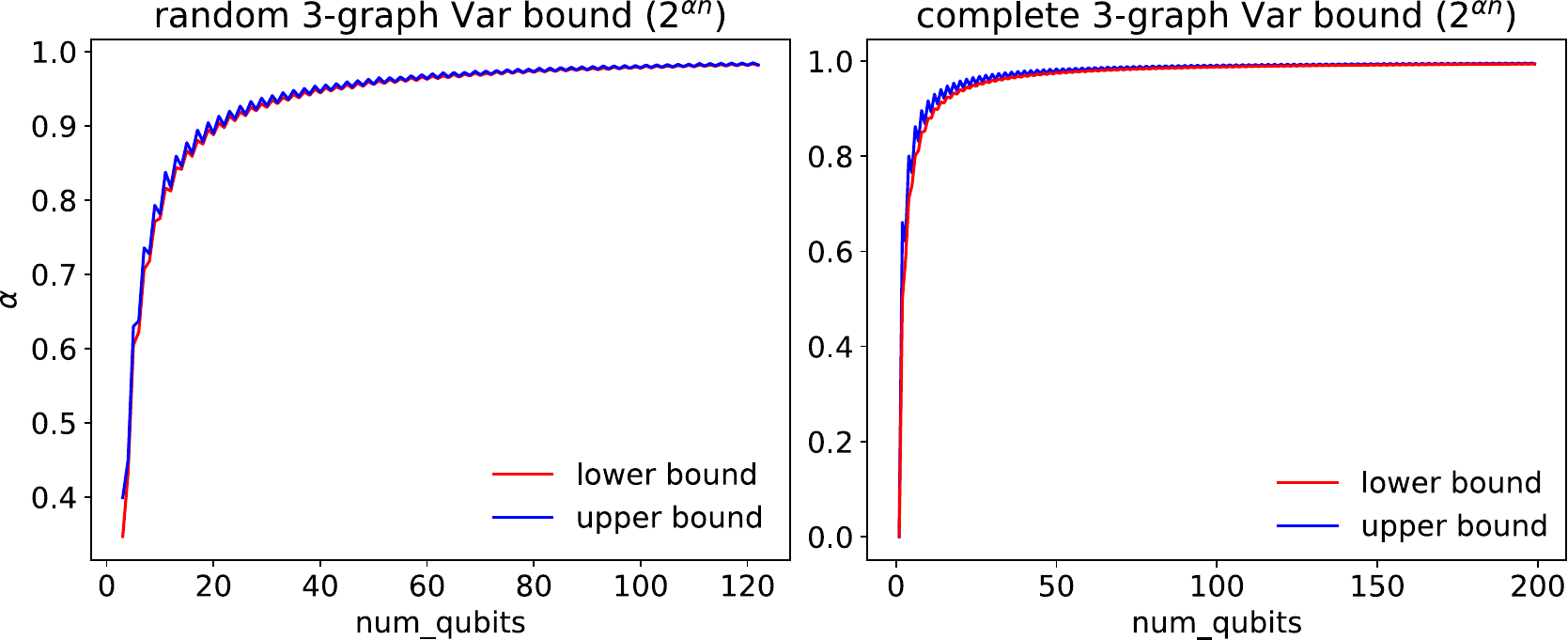}
	\caption{Upper and lower bounds of the variance of $\frac{1}{2}$-DFE for random- and complete-third-ordered hypergraph states. First graph used Eq.~\eqref{eq:var_graph_bound} with $2000$-random copies, and the second used Eq.~\eqref{eq:complete_graph_var}.} \label{fig:random_hypergraph_state_MSE_bound}
	
\end{figure}

The problem is reduced to finding the rank distribution of random hollow-symmetric matrices. To do so, we need the following known result, 
\begin{proposition}~\cite{macwilliams1977}\label{prop:macwilliams}
	Let $N(n,r)$ be the number of hollow-symmetric binary (or symplectic) matrices having the binary rank $r$. Given $h\in\mathbb{N}\cup \left\{0\right\}$, the following holds. 
	\begin{align}\label{eq:prob_hollow_symmetric_rank}
		N(n,2h+1)=0,\;N(n,2h)=\prod_{i=1}^{h}\frac{2^{2i-2}}{2^{2i}-1}\cdot \prod_{i=0}^{2h-1}(2^{n-i}-1).
	\end{align}
\end{proposition}
We can see that such matrices cannot have an odd rank.

Starting from Eq.~\eqref{eq:linear_entropy_hypergrpah_bound} with the convexity of exponential function,
we can also set the upper bound as follows, including the previous lower bound after the calculation followed by Eq.~\eqref{eq:prob_hollow_symmetric_rank}, 
\begin{corollary}
	For the $n$-qubit complete $3$-hypergraph state, the estimation variance of $1$-DFE is,
	\begin{align}\label{eq:complete_graph_var}
		2^{\frac{1}{2}\sum_{h=0}^{\lfloor\frac{n}{2}\rfloor}2hr(n,h)}\le{\rm Var}\le \sum_{h=0}^{\lfloor\frac{n}{2}\rfloor}r(n,h)2^{2h},
	\end{align}
	where $r(n,h)\equiv N(n,2h)2^{-\frac{n(n-1)}{2}}$ (see Prop.~\ref{prop:macwilliams} for the definition of $N$).
\end{corollary}

Fig.~\ref{fig:random_hypergraph_state_MSE_bound} shows the upper and lower bounds of the averaged variance over uniformly random third-ordered hypergraph states and complete hypergraph states. We can see the two lines are very close to each other and converge to $1$. Therefore, for the third-ordered random (or complete) hypergraph state cases, the $\frac{1}{2}$-DFE sampling overhead is $\mathcal{O}(2^n)$.

\section{Fan-out-based fidelity estimation scheme I: Framework}

The next two sections show the complete and algorithmic procedure for the fan-out-based fidelity estimation (FOFE) scheme, hence proving Theorem~$1$ in the main text. 
First, we recall the definition of the $n$-qubit \emph{phase state} $\ket{\eta}\equiv D(\phi_\eta)\ket{+}^{\otimes n}$ where the assigned mapping $\phi_\eta:\bfF^n_2\rightarrow [0,2\pi]$ such that the diagonal gate (or unitary) $D(\phi_\eta)$ is defined as $D(\phi_{\eta})\equiv \sum_{\bfx\in \bfF_2^n}e^{i\phi_\eta(\bfx)}\ket{\bfx}\bra{\bfx}$. This can be rewritten as density matrix as follows, 
\begin{align}\label{eq:phase_state_expansion}
	\ket{\eta}\bra{\eta}= \frac{1}{2^n}\sum_{\bfa\in \bfF_2}D(\phi_\eta)X^{\bfa}D^{\dag}(\phi_\eta)=\frac{1}{2^n}\sum_{
		\bfa\in \bfF_2^n
	}X^{\bfa}(X^{\bfa}D(\phi_\eta)X^{\bfa })D^{\dag}(\phi_\eta). 
\end{align}
It is well-known that~\cite{bravyi2019, huangH2023} diagonal gate $D(\phi_\eta)$ is re-expressed as the Pauli-$Z$ Hamiltonian evolution, i.e., $D(\phi_\eta)=e^{i\sum_{\bfa\in \bfF_2^n}\kappa^{(\eta)}_\bfa Z^{\bfa}}$, where $\forall \kappa^{(\eta)}_{\bfa}\in \mathbb{R}$. It means that $X^{\bfb}D(\phi_\eta)X^{\bfb}=e^{i\sum_{\bfa\in \bfF_2^n}c^{(\eta)}X^{\bfb}Z^{\bfa}X^{\bfb}}=e^{i\sum_{\bfa\in \bfF_2^n}(-1)^{\bfa\cdot \bfb}\kappa^{(\eta)}_\bfa Z^{\bfa}}$ is another diagonal gate. Since product of two diagonal gates is diagonal, we can regard $(X^{\bfa}D(\phi_\eta)X^{\bfa })D^{\dag}(\phi_\eta)$ as some diagonal gate $D^{(\bfa)}(\phi_\eta)$. 

We should note that the same phenomenon happens even when we twirl $D(\phi_\eta)$ with an \emph{arbitrary} $P\in \mathcal{P}_n$ or $T_{\bfa}\in \mathcal{P}_n/ \bfF_4\;(\bfa\in \bfF_2^{2n})$, since its Pauli $Z$-part just commutes with the diagonal gate. The \emph{phase stripping} $\ket{\psi}=\sum_{\bfx\in \bfF_2^{n}}\xi_{\bfx}\ket{\bfx}\mapsto \ket{\breve{\psi}}=\sum_{\bfx\in \bfF_2^n}|\xi_{\bfx}|\ket{\bfx}$ implies that arbitrary pure state $\ket{\psi}$ is some diagonal operation $D(\phi_{\psi})$ to \emph{phase stripped state} $\ket{\breve{\psi}}$ that is $\ket{\psi}=D(\phi_\psi)\ket{\breve{\psi}}$. We also consider its Pauli expansion $\ket{\breve{\psi}}\bra{\breve{\psi}}=\sum_{\bfa\in \bfF_2^{2n}}\breve{c}_{\bfa}T_{\bfa}\;(\forall \breve{c}_{\bfa}\in \mathbb{R})$. From now on, if it is not necessary, we omit the argument $(\phi_\psi)$ under the specification of the mapping $\phi_\eta$. We can generalize Eq.~\eqref{eq:phase_state_expansion} and substitute to the fidelity $\braket{\psi|\rho|\psi}$ which reads for $\alpha\in [0,1]$,
\begin{align}\label{eq:main_objective_arbitrary}
	\tr{\rho D\ket{\breve{\psi}}\bra{\breve{\psi}}D^{\dag}(\phi_\eta)}=2^{(1-\alpha)(M_{\alpha}(\breve{\psi})+n)-\alpha n}\sum_{\bfa\in \bfF_2^{2n}}\frac{|\breve{c}_{\bfa}|^{2\alpha}\cdot |\breve{c}_{\bfa}|^{-2\alpha+1}}{\sum_{\bfb\in \bfF_2^{2n}}|\breve{c}_\bfb^{2\alpha}|}\sgn(\breve{c}_{\bfa}) \cdot \frac{1}{2}\left\{\tr{D^{(\bfa)}\rho T_{\bfa}}+\tr{T_{\bfa}\rho D^{(\bfa)\dag}}\right\},
\end{align}
where $M_{\alpha}(\cdot)$ is $\alpha$-SRE (See Def.~\ref{eq:def_SRE}). We also used the fact that the fidelity is a real value, and hence the last equation is the realization of the previous complex equation. We also used the equalities $(\tr{A})^{*}=\tr{A^{\dag}}$, and $\tr{AB}=\tr{BA}$ for the arbitrary matrices $A,B$.

Let us say there is an unbiased estimator  of $\mathbb{E}(\widehat{\lambda}^{\bfa})\equiv \frac{1}{2}\left\{\tr{D^{(\bfa)}\rho T_{\bfa}}+\tr{T_{\bfa}\rho D^{(\bfa)\dag}}\right\}$, $\widehat{\lambda}^{(\bfa)}=\widehat{\lambda}_1^{(\bfa)}+\widehat{\lambda}_2^{(\bfa)}$ equipped with the two independent random variable $\lambda_1,\lambda_2$ following the probability $p_1^{(\bfa)}(\lambda),p_2^{(\bfa)}(\lambda)$ respectively, where we also defined $p^{(\bfa)}(\lambda)\equiv p_1^{(\bfa)}(\lambda)p_2^{(\bfa)}(\lambda)$. The proof of the existence of such an estimator is deferred to the next section. Along with those notations, we can promote the estimator to the unbiased estimator of $\tr{\rho D(\phi_\eta)\ket{\breve{\psi}}\bra{\breve{\psi}} D^{\dag}(\phi_\eta)}$ as follows. We sample $\bfa\in \bfF_2^{2n}$ from the distribution $\left\{\frac{|\breve{c}_{\bfa}|^{2\alpha}}{\sum_{\bfb\in \bfF_2^{2n}}|\breve{c}_\bfb^{2\alpha}|}\right\}_{\bfa\in \bfF_{2}^{2n}}$. Then we take the estimator as $2^{(1-\alpha)(M_{\alpha}(\breve{\psi})+n)-\alpha n} |\breve{c}_{\bfa}|^{-2\alpha+1} \sgn(c_{\bfa})\widehat{\lambda}^{(\bfa)}$. We leave the proof of unbiased estimation to readers. The estimation variance becomes, 
\begin{align}\label{eq:general_variance}
	{\rm Var}(\rho;\phi_\psi,\sigma)&\le\sum_{\bfa\in \bfF_2^{2n}}\frac{|\breve{c}_{\bfa}|^{2\alpha}\cdot |\breve{c}_{\bfa}|^{-4\alpha+2}}{\sum_{\bfb\in \bfF_2^{2n}}|\breve{c}_\bfb^{2\alpha}|}\sum_{\lambda}p^{(\bfa)}(\lambda)2^{2(1-\alpha)(M_{\alpha}(\breve{\psi})+n)-2\alpha n}\left(\widehat{\lambda}^{(\bfa)}\right)^2\nonumber\\&= \sum_{\bfa\in \bfF_2^{2n}}\frac{|\breve{c}_{\bfa}|^{2\alpha}\cdot |\breve{c}_{\bfa}|^{-4\alpha+2}}{\sum_{\bfb\in \bfF_2^{2n}}|\breve{c}_\bfb^{2\alpha}|}\sum_{\lambda}p^{(\bfa)}(\lambda)\left\{\left(\widehat{\lambda}^{(\bfa)}-\mathbb{E}(\widehat{\lambda}^{\bfa})\right)^2+\mathbb{E}(\widehat{\lambda}^{\bfa})^2\right\}2^{2(1-\alpha)(M_{\alpha}(\breve{\psi})+n)-2\alpha n}\nonumber\\&= 2^{(1-\alpha)(M_{\alpha}(\breve{\psi})+n)-\alpha n+1}\sum_{\bfa\in \bfF_2^{2n}}|\breve{c}_{\bfa}|^{-2\alpha+2}\left(\sum_{\lambda}p^{(\bfa)}(\lambda)\left\{\left(\widehat{\lambda}^{(\bfa)}-\mathbb{E}(\widehat{\lambda}^{\bfa})\right)^2+\mathbb{E}(\widehat{\lambda}^{\bfa})^2\right\}\right)\nonumber\\&=2\cdot 2^{(1-\alpha)M_{\alpha}(\breve{\psi})+\alpha M_{1-\alpha}(\breve{\psi})}
\end{align}
where we used the fact that $\forall \bfa\in \bfF_2^{2n}$, $  {\rm Var}(\rho;\phi_\eta,\bfa)\equiv\sum_{\lambda}p^{(\bfa)}(\lambda)\left(\widehat{\lambda}^{(\bfa)}-\mathbb{E}(\widehat{\lambda}^{\bfa})\right)^2\le 1$, which will be shown in the next section, and $\mathbb{E}(\widehat{\lambda}^{\bfa})\le \frac{1}{2}\max_{\psi:{\rm pure}}|\braket{\psi|T_{\bfa}D^{(\bfa)}+D^{(\bfa)\dag}T_{\bfa}|\psi}|\le 1$. The last expression becomes minimal when $\alpha=\frac{1}{2}$ ($\because$ Cor.~\ref{Cor:renyi_convex}). The scale factor $2$ is a rough bound. For instance, if the target state is a hypergraph state~\cite{morimae2017}, $|\widehat{\lambda}_{\bfa}|=1$.

$2^{(1-\alpha)M_{\alpha}(\breve{\psi})+\alpha M_{1-\alpha}(\breve{\psi})}$ is the core factor quantifying the estimation variance, and hence sampling complexity. Indeed this is because, by the Hoeffding inequality with the median of mean (MOM) estimation technique~\cite{Jerrum:1986random}, the required sampling copies to achieve the additive $\epsilon(\in (0,1])$-error is $\mathcal{O}\left(\frac{{2^{(1-\alpha)M_{\alpha}(\breve{\psi})+\alpha M_{1-\alpha}(\breve{\psi})}}}{\epsilon^2}\log\left(\frac{1}{\delta_f}\right)\right)$ with the failure probability $\delta_f\in (0,1]$. If $\alpha\rightarrow 1^{-}$, the sampling complexity becomes $\mathcal{O}\left(\frac{{2^{M_0(\breve{\psi})}}}{\epsilon^2}\log\left(\frac{1}{\delta_f}\right)\right)$. If $\alpha=\frac{1}{2}$, then it becomes $\mathcal{O}\left(\frac{{2^{M_1(\breve{\psi})}}}{\epsilon^2}\log\left(\frac{1}{\delta_f}\right)\right)=\mathcal{O}\left(\frac{{\|\psi\|_1^2}}{\epsilon^2}\log\left(\frac{1}{\delta_f}\right)\right)$.  In the next section, we genuinely show that such an estimator $\widehat{\lambda}$ exists and requires only $n$  $CNOT$-gates (\emph{fan-out gate}) with a single ancilla qubit.

The above arguments assume that we can efficiently sample the phase point $\bfa$ from the distribution {\small $\left\{\frac{|\breve{c}_{\bfa}|^{2\alpha}}{\sum_{\bfb\in \bfF_2^{2n}}|\breve{c}_\bfb^{2\alpha}|}\right\}_{\bfa\in \bfF_{2}^{2n}}$} and calculate $\sgn(\breve{c}_\bfa)$. Phase states, and Dicke states twirled by some diagonal gate, $D(\psi)\ket{{\rm Dic}(n,k)}$~($\because$ Prop.~\ref{prop:dicke_efficiency}), as the targets satisfy such a condition. In particular, for these cases, the Pauli group is partitioned by a small number of exponentially large collections of operators sharing the same coefficient. 
We could find other cases of target states satisfying the condition, while leaving it as a future problem. 

Given that the phase-stripped state can be easily implemented, we can apply the Bell sampling technique for efficient $ l_2$-sampling. Consider more general cases of target state, $\ket{\psi}=D(\phi)\bar{U}\ket{\mathbf{0}}$, where $D(\phi)$ is a diagonal gate and $\bar{U}$ is real-valued (orthogonal) unitary so that $\bar{U}\ket{\mathbf{0}}$ is still real in computational bases. There are the case where $\bar{U}$ is much easier to implement than $D(\phi)$, which hence occupies most of the magic of the target state. Some examples are when $\bar{U}$ is a product unitary, $\mathcal{O}(\log(n))$-unitary blocks~\cite{cerezo2020,vidal2003} or it generates some sparse and real-valued states~\cite{bartschi2022,zhang2022_qsp}. Furthermore, with these assumptions, $\sgn(\breve{c}_\bfa)$ is classically and efficiently calculated. The \emph{Bell sampling}~\cite{montanaro2017_ls} is formalized as,
\begin{proposition}[Bell-sampling]\cite{montanaro2017_ls}
	Given a target state $\ket{\psi}$, suppose we can prepare $\ket{\breve{\psi}}$ as input copies. Using $n$-CNOTs in $2n$-qubit system, we can sample $\bfa$ by the Born probability $\left\{\frac{
		\braket{\breve{\psi}|T_{\bfa}|\breve{\psi}}
		^2}{2^n}\right\}$ ($l_2$-sampling).
\end{proposition}
\begin{proof}
	We prepare two copies of state $\ket{\breve{\psi}}^{\otimes 2}$, then enact $\prod_{i=1}^n CNOT_{i,i+n}(H^{\otimes n}\otimes I^{\otimes n})$. Finally, we take the computational basis measurement to the whole qubits, obtaining the output $\bfb=(\bfb_1,\bfb_2)\in \bfF_2^{2n}$. The following Born probability is, denoting $2n$-qubit Bell state (or maximally entangled state) as $\ket{\Phi}=\frac{1}{\sqrt{2^n}}\sum_{\bfx\in \bfF_2^n}\ket{\bfx\bfx}$,
	\begin{align}
		&\tr{(Z^{\bfb_1}\otimes X^{\bfb_2})\ket{\breve{\psi}}\bra{\breve{\psi}}^{\otimes 2}(Z^{\bfb_1}\otimes X^{\bfb_2})\ket{\Phi}\bra{\Phi}}\nonumber\\&=\frac{1}{2^n}\tr{Z^{\bfb_1}\ket{\breve{\psi}}\bra{\breve{\psi}}Z^{\bfb_1}\left(X^{\bfb_2}\ket{\breve{\psi}}\bra{\breve{\psi}}X^{\bfb_2}\right)^{\top}}=\frac{1}{2^n}\tr{Z^{\bfb_1}\ket{\breve{\psi}}\bra{\breve{\psi}}Z^{\bfb_1}X^{\bfb_2}\ket{\breve{\psi}^{\ast}}\bra{\breve{\psi}^{\ast}}X^{\bfb_2}}=\frac{1}{2^n}|\braket{\breve{\psi}|T_{(\bfb_2,\bfb_1)}|\breve{\psi^{\ast}}}|^2\nonumber\\&=\frac{1}{2^n}\braket{\breve{\psi}|T_{(\bfb_2,\bfb_1)}|\breve{\psi}}^2,
	\end{align}
	, where $\ket{\breve{\psi}^{\ast}}$ is the complex conjugation of the whole coefficients with the computational basis, making no difference. Finally, sampled $\bfa$ is $(\bfb_2,\bfb_1)$.
\end{proof}

Furthermore, we note that the Bell measurement circuit is itself a Clifford circuit. If the unitary generating $\ket{\breve{\psi}}$ is dominated by Clifford gates and low-T gates, then the $l_2$-sampling can also be classically simulated~\cite{bravyi2016i,bravyi2019,qassim2021}, without preparing $\ket{\breve{\psi}}$ as an input.

\begin{figure}[t]
	\includegraphics[width=\columnwidth]{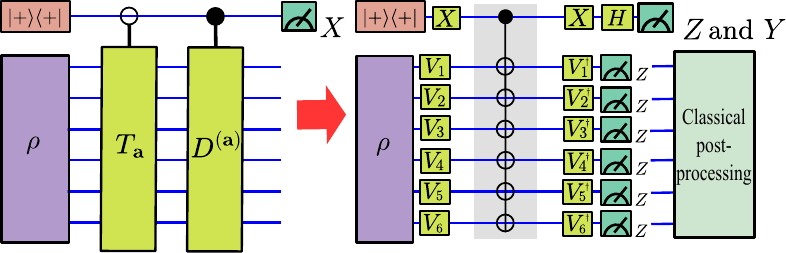}
	\caption{Schematic illustration of FOFE.}\label{fig:diagonal_state_fidelity_scheme}
	
\end{figure}

\section{Fan-out-based fidelity estimation scheme II: Hadamard test circuit and post-processing}

What is the quantum algorithm that estimates $\mathbb{E}(\widehat{\lambda}^{\bfa})\equiv \frac{1}{2}\left\{\tr{D^{(\bfa)}\rho T_{\bfa}}+\tr{T_{\bfa}\rho D^{(\bfa)\dag}}\right\}$ of Eq.~\eqref{eq:main_objective_arbitrary}? We can apply some special quantum circuit, which we refer as \emph{Hadamard test}~\cite{faehrmann2025}. This technique was referred from the previous works~\cite{sun2022,tsubouchi2023_virtual,faehrmann2025} that led to other applications. Hadamard test is used to estimate the expectation value of the observable $O$ after the non-physical simulation from unitaries $V$ and $U$, 
\begin{align}
	\braket{O}_{(U,V)}\equiv\frac{1}{2}\left\{\tr{OU\rho V^{\dag}}+\tr{OV\rho U^{\dag}}\right\}.
\end{align}
The above process is non-physical because such conjugation $\rho\rightarrow \frac{1}{2}\left(U\rho V^{\dag}+V\rho U^{\dag}\right)$ is not the mapping from an existing quantum channel. Nevertheless, it is important to note that we require the $\braket{O}_{(U,V)}$ of such a non-physical output, rather than its full information. Then the problem might be easier. Let us explain how to achieve that estimation. First, we prepare an additional ancilla state $\ket{+}\bra{+}$. Then we give operations inverted-controlled-$U$ ($C_0 U$) and controlled-$V$ ($CV$) subsequently, where the ancilla will take part of the control qubit, and \emph{inverted} means that the target operation is activated for $0$-control-qubit, hence the output becomes, 
\begin{align}
	CVC_0U\ket{+}\bra{+}\otimes\rho\; C_0U^{\dag}CV^{\dag}=\frac{1}{2}\left\{\ket{0}\bra{0}\otimes U\rho U^{\dag}+\ket{0}\bra{1}\otimes U\rho V^{\dag}+\ket{1}\bra{0}V\rho U^{\dag}+\ket{1}\bra{1}V\rho V^{\dag}\right\}.
\end{align}
Then the expectation value of $X\otimes O$ is, 
\begin{align}
	\tr{CVC_0U\ket{+}\bra{+}\otimes\rho\; C_0U^{\dag}CV^{\dag}(X\otimes O)}=\frac{1}{2}\left\{\tr{OU\rho V^{\dag}}+\tr{OV\rho U^{\dag}}\right\},
\end{align}
which is the desired result. Eq.~\eqref{eq:main_objective_arbitrary} is a corollary from the substitution, $O=I$, $U=T_{\bfa}$, and $V=D^{(\bfa)}$. That is, 
\begin{align}
	\tr{(X\otimes I)CD^{(\bfa)}C_0T_{\bfa}(\ket{+}\bra{+}\otimes \rho)C_0T_{\bfa}CD^{(\bfa)\dag}}=\frac{1}{2}\left\{\tr{D^{(\bfa)}\rho T_{\bfa}}+\tr{T_{\bfa}\rho D^{(\bfa)\dag}}\right\},
\end{align}
where the control gates are defined as $C_0T_{\bfa}\equiv \ket{0}\bra{0}\otimes T_{\bfa}+\ket{1}\bra{1}\otimes I$ and $CD^{(\bfa)}\equiv \ket{0}\bra{0}\otimes I+\ket{1}\bra{1}\otimes D^{(\bfa)}$.

Therefore, our circuit seems to require, as entangling gates, $n$ number of CNOT gates and gates for $CD^{(\bfa)}$, which requires in general $\mathcal{O}(2^n)$ Clifford+T gates~\cite{gosset2025}. However, we can transform such a complex operation to the post-processing of the Pauli measurement outcomes so that we do not need such an excessive gate overhead. To achieve this, suppose the output state after the $C_0T_{\bfa}$ with the additional $(H\otimes I)$ operation (see the right picture in Fig.~\ref{fig:diagonal_state_fidelity_scheme}) has a spectral decomposition $\rho_{\rm output}=\sum_{\tau}\tau\ket{\tau}\bra{\tau}$. Without losing of generality, we only pick one eigenvector $\ket{\tau}=\sum_{x_1\in \bfF_2,\bfx\in \bfF_n^2}\xi_{\bfx}\ket{x_1,\bfx}$ as the input for remaining gates $D^{(\bfa)}$ and estimation of $(X\otimes I)$. We know the following representation,
\begin{align}
	(H\otimes I)CD^{(\bfa)}(H\otimes I)\ket{x_1,\bfx}&=\ket{+}\bra{+}\otimes I+\ket{-}\bra{-}\otimes D^{(\bfa)}\ket{x_1,\bfx}\nonumber\\&=\frac{1+(-1)^{x_1}e^{i\phi^{(\bfa)}_\psi(\bfx)}}{2}\ket{0,\bfx}+\frac{1-(-1)^{x_1}e^{i\phi^{(\bfa)}_\psi(\bfx)}}{2}\ket{1,\bfx}.
\end{align}
This fact rewrites the expectation value to, 
\begin{align}\label{eq:twirled_and_measure}
	&\braket{\tau| (H\otimes I)CD^{(\bfa)}(H\otimes I)(Z\otimes I) (H\otimes I)CD^{(\bfa)\dag}(H\otimes I)|\tau}\nonumber\\&=\sum_{\substack{x_1,y_1\in F_2\\\bfx,\bfy\in \bfF_2^n}}\xi_{(x_1,\bfx)}^*\xi_{(y_1,\bfy)}\delta_{\bfx,\bfy}\left\{\frac{1+(-1)^{y_1}e^{i\phi^{(\bfa)}_\psi(\bfy)}+(-1)^{x_1}e^{-i\phi^{(\bfa)}_\psi(\bfx)}+(-1)^{x_1+y_1}e^{i\phi^{(\bfa)}_\psi(\bfy)-i\phi^{(\bfa)}_\psi(\bfx)}}{4}\right.\nonumber\\&\left.-\frac{1-(-1)^{y_1}e^{i\phi^{(\bfa)}_\psi(\bfy)}-(-1)^{x_1}e^{-i\phi^{(\bfa)}_\psi(\bfx)}+(-1)^{x_1+y_1}e^{i\phi^{(\bfa)}_\psi(\bfy)-i\phi^{(\bfa)}_\psi(\bfx)}}{4}\right\}\nonumber\\&=\sum_{x_1,y_1\in \bfF_2,\bfx\in \bfF_2^n}\xi_{(x_1,\bfx)}^*\xi_{(y_1,\bfx)}\left\{\frac{(-1)^{y_1}e^{i\phi^{(\bfa)}_\psi(\bfx)}+(-1)^{x_1}e^{-i\phi^{(\bfa)}_\psi(\bfx)}}{2}\right\}\nonumber\\&=\sum_{(x_1,\bfx)\in \bfF_2^{n+1}}|\xi_{x_1,\bfx}|^2(-1)^{x_1}\left\{\frac{e^{i\phi^{(\bfa)}_\psi(\bfx)}+e^{-i\phi^{(\bfa)}_\psi(\bfx)}}{2}\right\}+\sum_{\bfx\in \bfF_2^n}\left(-i\xi_{(0,\bfx)}\xi^{*}_{(1,\bfx)}+i\xi_{(1,\bfx)}\xi^{*}_{(0,\bfx)}\right)\left\{\frac{e^{i\phi^{(\bfa)}_\psi(\bfx)}-e^{-i\phi^{(\bfa)}_\psi(\bfx)}}{2i}\right\}\nonumber\\&=\braket{\tau|Z\otimes {\rm Re}(D^{(\bfa)})|\tau }-\braket{\tau|Y\otimes {\rm Im}(D^{(\bfa)})|\tau }.
\end{align}
Recall the spectral decomposition of the output state $\rho_{\rm output}=\sum_{\tau}\tau\ket{\tau}\bra{\tau}$. Using Eq.~\eqref{eq:twirled_and_measure}, finally,
\begin{align}
	\tr{(X\otimes I)CD^{(\bfa)}C_0T_{\bfa}(\ket{+}\bra{+}\otimes \rho)C_0T_{\bfa}CD^{(\bfa)\dag}}=\tr{\rho_{\rm output}(Z\otimes {\rm Re}(D^{(\bfa)}))}-\tr{\rho_{\rm output}(Y\otimes {\rm Im}(D^{(\bfa)}))}
\end{align}

In conclusion, we can estimate the desired expectation value $\tr{(X\otimes I)CD^{(\bfa)}C_0T_{\bfa}(\ket{+}\bra{+}\otimes \rho)C_0T_{\bfa}CD^{(\bfa)\dag}}=\frac{1}{2}\left\{\tr{D^{(\bfa)}\rho X^{\bfa}}+\tr{X^{\bfa}\rho D^{(\bfa)\dag}}\right\}$ as follows: We prepare additional ancilla state $\ket{+}\bra{+}$ and then we enact the gates $(H\otimes I)C_0T_{\bfa}$. Next, we estimate the expectation value $\braket{Z\otimes {\rm Re}(D^{(\bfa)})}$ of the output state. This is possible by doing the computational basis measurement to obtain the outcome $(b_1,\bfb)\in \bfF_2^{n+1}$ and take the estimator $\widehat{\braket{Z\otimes {\rm Re}(D^{(\bfa)})}}=(-1)^{b_1}\cos(\phi^{(\bfa)}_\psi(\bfb))$. Then we prepare additional copies enacted by $(H\otimes I)C_0T_{\bfa}$ to estimate $\braket{Y\otimes {\rm Im}(D^{(\bfa)})}$, which is possible by measuring in the computational basis except for the first qubit that is measured in $Y$-basis to obtain $(b'_1,\bfb')\in \bfF_2^{n+1}$. Then we take the estimator $-\widehat{\braket{Y\otimes {\rm Im}(D^{(\bfa)})}}=(-1)^{b'_1+1}\sin(\phi^{(\bfa)}_\psi(\bfb'))$. Along with the knowledge of the previous section, we finally get the unbiased estimator of $\braket{\eta|\rho|\eta}$, requiring only $n$-CNOTs as entangling gates and one ancilla qubit. We remember tha $\phi^{(\bfa)}_{\psi}$ is the phase function for $D^{(\bfa)}$. Since $\forall \bfx\in \bfF_2^n,\;D^{{(\bfa)}}\ket{\bfx}=T_{\bfa}DT_{\bfa}D^{\dag}\ket{\bfx}=e^{i\left(\phi_\psi(\bfx+\bfa_x)-\phi_{\psi}(\bfx)\right)}\ket{\bfx}$, it leads to $\phi^{(\bfa)}_\psi(\bfx)=\phi_\psi(\bfx+\bfa_x)-\phi_\psi(\bfx)$, which can be efficiently calculated if $\phi^{(
	\bfa)}_\psi(\bfx)$ is assumed to be efficiently found. 

The only thing left is the estimation variance. Let us denote the $n+1$-qubit output state before the measurement is $\rho^{(\bfa)}$. We start from the result of Eq.~\eqref{eq:twirled_and_measure} and the notation of the previous section. Since the estimators of $\braket{Z\otimes {\rm Re}(D^{'(\bfa)})}$ and $\braket{Y\otimes {\rm Im}(D^{'(\bfa)})}$ have independent samplers, the total variance is summed. In general, given random variables $X,Y$, ${\rm Var}(X+Y)={\rm Var}(X)+{\rm Var}(Y)+{
	\rm Cov
}(X,Y)$, in which the independent sampling leads to ${
	\rm Cov
}(X,Y)=0$. More specifically,

\begin{align}
	{\rm Var}(\rho;\phi_\eta,\bfa)&= \sum_{(x_1,\bfx)\in \bfF_2^{n+1}}\tr{\ket{x_1}\bra{x_1}\otimes \ket{\bfx}\bra{\bfx}\rho^{(\bfa)}}\cos^2(\phi(\bfx))-\braket{Z\otimes {\rm Re}(D^{'(\bfa)})}^2\nonumber\\&+\sum_{(x_1,\bfx)\in \bfF_2^{n+1}}\tr{\ket{(-1)^{x_1}i}\bra{(-1)^{x_1}i}\otimes \ket{\bfx}\bra{\bfx}\rho^{(\bfa)}}\sin^2(\phi(\bfx))-\braket{Y\otimes {\rm Im}(D^{'(\bfa)})}^2\nonumber\\&\le\sum_{\bfx\in \bfF_2^{n}}\tr{I\otimes \ket{\bfx}\bra{\bfx}\rho^{(\bfa)}}\cos^2(\phi(\bfx))+\sum_{\bfx\in \bfF_2^{n}}\tr{I\otimes \ket{\bfx}\bra{\bfx}\rho^{(\bfa)}}\sin^2(\phi(\bfx))\nonumber\\&=\sum_{\bfx\in \bfF_2^n}\braket{\bfx|\mathrm{tr}_1\{\rho^{(\bfa)}\}|\bfx}=1,
\end{align}
where we used the total variance is the sum of each variance of the independent estimator. Finally, we confirmed Eq.~\eqref{eq:general_variance}, ${\rm Var}(\rho;\phi_\eta,\bfa)\le 2\cdot 2^{(1-\alpha)M_{\alpha}(\breve{\psi})+\alpha M_{1-\alpha}(\breve{\psi})}$. 

The next proposition implies an additional benefit in the case where the magic (non-stabilizerness) of $\ket{\psi}$ is totally endowed by the diagonal gate. 
\begin{proposition}
	If the target state $\ket{\psi}=D(\phi_{\psi})\ket{\omega}$ is a stabilizer state $\ket{\omega}$ acted by some diagonal operator $D(\phi_\psi)$,  then $\|\breve{\psi}\|_1=\|\breve{\omega}\|_1=1$.
\end{proposition}
\begin{proof}
	Every stabilizer state $\ket{\omega}$ can be expressed as follows~\cite{bravyi2016i,bravyi2016}, 
	\begin{align}
		\ket{\omega}=\frac{1}{\sqrt{2^{\dim{A}}}}\sum_{\bfx\in A}i^{\bfu\cdot \bfx}(-1)^{Q(\bfx)}\ket{\bfx+\bfu},
	\end{align}
	where $A$ is some vector subspace in $\bfF_2^n$, $\bfu\in \bfF_2^n$ is fixed vector such that $\bfu\cdot \bfx$ is calculated in modular $4$, and $Q$ is some second-degree Boolean function. We also remember that the diagonal gate does not change the magnitude of the coefficients. In light of this knowledge, we conclude that 
	\begin{align}
		\|\breve{\psi}\|_1=\|\bar{\omega}\|_1=1\;{\rm where}\; \ket{\bar{w}}=\frac{1}{\sqrt{2^{\dim{A}}}}\sum_{\bfx\in A}\ket{\bfx+\bfu},
	\end{align}
	since $\ket{\bar{\omega}}$ is another stabilzer state that can be generated via $X$, $H$ and $CNOT$ gates to $\ket{0}^{\otimes n}$. 
\end{proof}

Therefore, the following corollary is that estimating the fidelity with state which has the following form $\ket{\psi}=D(\phi_\psi)\ket{\omega}$, where $\ket{\omega}$ is the stabilizer state, can be efficiently done with $\mathcal{O}\left(\epsilon^{-2}\log(\delta_f^{-1})\right)$ (constant) number of samplings with our scheme because in this case, $\|\breve{\psi}\|_1=1$. We call such $\ket{\psi}$ a \emph{phase-stabilizer state}. If $\ket{\omega}=\ket{+}^{\otimes n}$, it shrinks to a \emph{phase state}~\cite{lee2025_shallow}. 


Importantly, let us consider when we estimate the fidelities with $M$-number of phase states $\left\{\ket{\eta_1},\ket{\eta_2},\ldots,\ket{\eta_M}\right\}$. We remember that all phase states share the same $l_{2\alpha}$-sampling for FOFE, the uniform Pauli $X$-operators. It means that we use the same measurement circuit for all phase states. Consequently, after getting the measurement outcome $\bfb$ we can use to calculate \emph{many} estimators $(-1)^{b_1}\cos(\phi^{(\bfa)}_\eta(\bfb))$ or $(-1)^{b_1+1}\sin(\phi^{(\bfa)}_\eta(\bfb))$ following the structure of each $\eta$. Therefore, we just need to re-scale the failure probability for each phase state to $M^{-1}\delta_f$. Then Thm.~\ref{main:thm_1} is proved. In other words, the number of target phase states merely gives the log-factor of the sampling complexity. We can also directly generalize this property. If all $M$ target states $\left\{\ket{\psi_1},\ket{\psi_2},\ldots,\ket{\psi_M}\right\}$ share the same phase stripped state $\ket{\breve{\psi}}$, the required sampling complexity for FOFE is $\mathcal{O}\left(\frac{\|\breve{\psi}\|_{2-2\alpha}^{1/\alpha}}{\epsilon^2}\log(M\delta_f^{-1})\right)$. This reasoning is analogous to the main virtue of the classical shadow~\cite{huang2020,aronson2018}: the measurement outcome followed by classical shadow $\mathcal{M}^{-1}(\ket{\bfb}\bra{\bfb})$ ($\mathcal{M}^{-1}$: inversion of measurement channel~\cite{chen2021_rs}) is used to estimate expectation values of \emph{many} observables.

\section{Fidelity estimation overhead of random phase-stripped states}

We recall that given an arbitrary pure state $\ket{\psi}=\sum_{\bfx\in \bfF_2^n}\xi_{\bfx}\ket{\bfx}$, we denote $\ket{\breve{\psi}}=\sum_{\bfx}|\xi_{\bfx}|\ket{\bfx}$ as its phase-stripped state. In this section, we calculate the average of $\|\breve{\psi}\|_1$ over the Haar random states $\ket{\psi}$, comparing with the average of $\|\psi\|_1$.

To do so, we need to review basic properties of the Dirichlet function~\cite{kotz2019}. Let us fix $K\in \mathbb{N}\backslash\left\{1\right\}$, and $\alpha_1,\alpha_2,\ldots,\alpha_K>0\;(\alpha\equiv (\alpha_1,\alpha_2,\ldots,\alpha_K))$. The Dirichlet distribution is defined as the probability distribution over a $K$-sized probability simplex $\triangle_K$ as follows, 
\begin{align}
	{\rm Dir}(p\equiv(p_1,p_2,\ldots,p_K);\alpha)\equiv \frac{\prod_{i=1}^{K}\Gamma(\alpha_i)}{\Gamma\left(\sum_{i=1}^K\alpha_i\right)}\times\prod_{i=1}^{K}p^{\alpha_i-1}_i,
\end{align}
where $\Gamma(x)\equiv \int_{0}^{\infty}t^{x-1}e^{-t}dt$ is gamma function. Uniform sampling corresponds to when $\alpha=(1,1,1,\ldots,1)$. Dirichlet distribution satisfies the following moment rules, 

\begin{align}\label{eq:Dirichlet_moment}
	\int_{\triangle_K} d^{(\alpha)}p\left(\prod_{i=1}^{K}X_i^{\beta_i}\right)=\int_{\triangle_K} J_{\alpha}(p_1,p_2,\ldots,p_K)dp_1dp_2\ldots dp_{K}\left(\prod_{i=1}^{K}X_i^{\beta_i}\right)=\frac{\Gamma\left(\sum_{i=1}^K\alpha_i\right)}{\Gamma\left(\sum_{i=1}^K(\alpha_i+\beta_i)\right)}\times \prod_{i=1}^{K}\frac{\Gamma(\alpha_i+\beta_i)}{\Gamma(\alpha_i)},
\end{align}
where $d^{(\alpha)}p=J_{\alpha}(p_1,p_2,\ldots,p_K)dp_1dp_2\ldots p_K$ is the integral measure from ${\rm Dir}(p;\alpha)$. Second, we have the following marginal distribution ${\rm Dir}_L\equiv{\rm Dir}(p_1+p_2+\ldots+p_L)\;(L\in \mathbb{N},L<K)$ that is $\forall i\in [K]$, 
\begin{align}\label{eq:Dirichlet_marginal}
	{\rm Dir}(p_1+p_2+\ldots+p_L)=\frac{\Gamma(\alpha_0)}{\Gamma(\alpha_L)\Gamma(\alpha_0-\alpha_L)}\times p_i^{\alpha_L-1}(1-p_i)^{\alpha_0-\alpha_L-1}\;(\alpha_0\equiv \sum_{i=1}^{K}\alpha_i,\; \alpha_L\equiv\sum_{i=1}^{L}\alpha_i)
\end{align}

Throughout this paper, we shall fix $K=2^n$, $\alpha=(1,1,1,\ldots,1)$ and let $p=(p_{\bfx})_{\bfx\in \bfF_2^n}$. Haar random~\cite{mele2023} refers to a uniform measure over the random pure states, that is, we sample each coefficient $\xi_{\bfx}$ independently from the normal distribution $\mathcal{N}_{\mathbb{C}}(0,1)$ then normalize so that $\sum_{\bfx}|\xi_{\bfx}|^2=1$. This is equivalent to randomly sampling the probability distribution $p\in \triangle_{2^n}$  following ${\rm Dir}(1,1,\ldots,1)$, and then randomly sampling each $\phi_\bfx\in [0,2\pi]$ and finally setting $\forall \bfx,\;\xi_{\bfx}=\sqrt{p_{\bfx}}e^{i\phi_\bfx}$.  

First, we calculate the average of $\|\psi\|_1$ over Haar random states. We note that 
\begin{align}
	\hmoment\|\psi\|_1=\frac{1}{2^n}+\hmoment\sum_{\bfa\in \bfF_2^{2n},\bfa\ne \mathbf{0}}|\braket{\psi|T_{\bfa}|\psi}|=\frac{1}{2^n}+\frac{4^n-1}{2^n}\hmoment|\braket{\psi|ZIII\ldots I|\psi}|. 
\end{align}
The last equality is because all non-identity Pauli operator is equivalent under the conjugation by some Clifford operator~\cite{aaronson2004}, hence $  \mathbb{E}_{\psi\in {\rm Haar}}|\braket{\psi|P|\psi}|=\mathbb{E}_{\psi\in {\rm Haar}}|\braket{\psi|Q|\psi}|$ for arbitrary non-identity Pauli operators $P\ne Q$.

Further simplification of the last form is well-known~\cite{bengtsson2017}, where the detail is given as follows, 
\begin{align}\label{eq:Haar_random_l1_norm_integral}
	&\hmoment\|\psi\|_1\nonumber\\&=\frac{1}{2^n}+\frac{4^n-1}{2^n}\hmoment \left|\sum_{\bfx\in \bfF_2^n}|\xi_\bfx|^2(-1)^{x_1}\right|=\frac{1}{2^n}+\frac{4^n-1}{2^n}\hmoment \left|2\sum_{\bfx\in \bfF_2^n,x_1=0}|\xi_\bfx|^2-1\right|\nonumber\\&=\frac{1}{2^n}+\frac{4^n-1}{2^n}\int_{0}^{\frac{1}{2}}{\rm Dir}_{2^{n-1}}(P)(1-2P)dP+\frac{4^n-1}{2^n}\int_{\frac{1}{2}}^{1}{\rm Dir}_{2^{n-1}}(P)(2P-1)dP\nonumber\\&=\frac{1}{2^n}+\frac{4^n-1}{2^n}\int_{0}^{\frac{1}{2}}\frac{(2^n-1)!}{(2^{n-1}-1)!^2}(1-P)^{2^{n-1}-1}(P^{2^{n-1}-1}-2P^{2^{n-1}})dP\nonumber\\&+\frac{4^n-1}{2^n}\int_{\frac{1}{2}}^{1}\frac{(2^n-1)!}{(2^{n-1}-1)!^2}(1-P)^{2^{n-1}-1}(2P^{2^{n-1}}-P^{2^{n-1}-1})dP\nonumber\\&=\frac{1}{2^n}+\frac{4^n-1}{2^n}\times \frac{(2^n-1)!}{(2^{n-1}-1)!^2}\left(2\int_{0}^{\frac{1}{2}}-\int_0^1\right)(1-P)^{2^{n-1}-1}(P^{2^{n-1}-1}-2P^{2^{n-1}})dP\nonumber\\&=\frac{1}{2^n}+\frac{(4^n-1)(2^n-1)!}{2^n(2^{n-1}-1)!^2}\left\{2B\left(\frac{1}{2};2^{n-1},2^{n-1}\right)-4B\left(\frac{1}{2};2^{n-1}+1,2^{n-1}\right)-B(2^{n-1},2^{n-1})+2B(2^{n-1}+1,2^{n-1})\right\}\nonumber\\&\simeq \sqrt{\frac{2^{n+1}}{\pi}}\simeq 0.
	798\times 2^{0.5n}.
\end{align}
where the second inequality comes from Eq.~\eqref{eq:Dirichlet_marginal} and the fact that $\left\{\bfx\in \bfF_2^n|x_1=0\right\}$ is vector subspace of the size $2^{n-1}$. The last two inequalities is derived by the following arguments: 
We first ignore the first $\frac{1}{2^n}$ term. Next, we use Stirling's formula $k!\simeq \sqrt{2\pi k}\left(\frac{k}{e}\right)^k$ to obtain $\frac{(4^n-1)(2^n-1)!}{2^n(2^{n-1}-1)!^2}\simeq \sqrt{\frac{1}{8\pi}}\times 2^{2^n+\frac{3n}{2}}$. Finally, we use the definition of the incomplete beta function~\cite{kotz2019} with several properties, 
\begin{align}\label{eq:properties_incomplete_beta}
	B\left(x;a,b\right)
	\equiv \int_0^{x}t^{a-1}(1-t)^{b-1}dt\Rightarrow   \begin{cases}
		B(1;a,b)=B(a,b)\; \left(B(a,b)\equiv \frac{\Gamma(a)\Gamma(b)}{\Gamma(a+b)}\right)\\B(a+1,a)=\frac{1}{2}B(a,a)\\
		B\left(\frac{1}{2};a,a\right)=\frac{1}{2}B(a,a)\\ B\left(\frac{1}{2};a+1,a\right)=\frac{1}{4}B(a,a)-\frac{1}{a2^{2a+1}},
	\end{cases}
\end{align}
(see App. A for its proof) and hence,
\begin{align}
	&2B\left(\frac{1}{2};2^{n-1},2^{n-1}\right)-4B\left(\frac{1}{2};2^{n-1}+1,2^{n-1}\right)-B(2^{n-1},2^{n-1})+2B(2^{n-1}+1,2^{n-1})\nonumber\\&\simeq -4B\left(\frac{1}{2};2^{n-1}+1,2^{n-1}\right)+2B(2^{n-1}+1,2^{n-1})=\frac{4}{2^{n+2^n}} +\frac{1}{2}B(2^{n-1},2^{n-1})\simeq \frac{4}{2^{n+2^n}},
\end{align}
since $\frac{\Gamma(2^{n-1})^2}{\Gamma(2^n)}\simeq 0$.
We obtain one result, $\hmoment\|\psi\|_1\simeq \sqrt{\frac{2^{n+1}}{\pi}}$. Now, let us calculate $\hmoment\|\breve{\psi}\|_1$. To do so, we need a lemma,
\begin{lemma}
	For arbitrary $\bfa\in \bfF_2^{2n}$, $\hmoment|\braket{\breve{\psi}|T_{\bfa}|\breve{\psi}}|$ belongs to
	\begin{align}\label{eq:non_trivial_phase_strip_cases}
		\left\{0,\hmoment|\braket{\breve{\psi}|II\ldots I|\breve{\psi}}|=1,\hmoment|\braket{\breve{\psi}|ZI\ldots I|\breve{\psi}}|,\hmoment|\braket{\breve{\psi}|XI\ldots I|\breve{\psi}}|,\hmoment|\braket{\breve{\psi}|ZXI\ldots I|\breve{\psi}}|\right\}.
	\end{align}
\end{lemma}

\begin{proof}
	We first claim that if $P,Q\in \mathcal{P}_n$ are equivalent under $CNOT$ operations, say the mapping is $L\ket{\bfx}=\ket{\widetilde{L}\bfx}$ for a given linear map $\widetilde{L}:\bfF_2^n\rightarrow \bfF_2^n$, then $\hmoment|\braket{\breve{\psi}|P|\breve{\psi}}|=\hmoment|\braket{\breve{\psi}|Q|\breve{\psi}}|$. To explain the reason, we first note that $\ket{\psi}=\sum_{\bfx}\xi_{\bfx}\ket{\bfx}=\sum_{\bfx}|\xi_{\bfx}|e^{i\phi(\bfx)}\ket{\bfx}$ for some function $\phi:\bfF^n_2\rightarrow [0,2\pi]$ (phase). Hence, $\ket{\psi}=D(\phi)\ket{\breve{\psi}}$, where $D(\phi)=e^{i\sum_{\bfa\in \bfF_2^n}\kappa_{\bfa}Z^{\bfa}}$ and $\forall \kappa_{\bfa}\in \mathbb{R}$. We also note that $L^{-1}\ket{\psi}=\sum_{\bfx}|\xi_{\bfx}|L^{-1}D(\phi)L L^{-1}\ket{\bfx}=D'(\phi)L^{-1}\ket{\breve{\psi}}$, where $D'(\phi)\equiv e^{i\sum_{\bfa}\kappa_{\bfa}L^{-1}Z^{\bfa}L}=e^{i\sum_{\bfa}\kappa_{\bfa}Z^{\widetilde{L}(\bfa)}}$ which is another diagonal gate. Moreover, we note that $L^{-1}\ket{\psi}=\sum_{\bfx}\xi_{\bfx}\ket{\widetilde{L}\bfx}=\sum_{\bfx}\xi_{\widetilde{L}^{-1}\bfx}\ket{\bfx}$, hence using previous equations, $\breve{L^{-1}\ket{\psi}}=\sum_{\bfx}|\xi_{\widetilde{L}^{-1}\bfx}|\ket{\bfx}=D^{'\dag}(\phi)L^{-1}\ket{\phi}$ and we conclude $\breve{L^{-1}\ket{\psi}}=L^{-1}\ket{\breve{\psi}}$. Therefore, 
	\begin{align}
		\hmoment|\braket{\breve{\psi}|Q|\breve{\psi}}|=\hmoment|\braket{\breve{\psi}|LPL^{-1}|\breve{\psi}}|=\hmoment|\braket{\breve{L^{-1}\psi}|P|\breve{L^{-1}\psi}}|=\hmoment|\braket{\breve{\psi}|P|\breve{\psi}}|,
	\end{align}
	where the last equation is obtained by left invariance of Haar random. We proved the claim.
	
	Let us get back on track. Without losing generality, we suppose a Pauli operator with the following $n$-sized string,
	\begin{align}
		T_{\bfa}=XXX\ldots XYYY\ldots YZZZ\ldots ZIII\ldots I
	\end{align}
	By the previous claim, other qubit-shuffled cases are equivalent by the $SWAP$ operations, which are formed by $CNOT$s. We also use the fact that 
	\begin{align}
		CNOT_{1\rightarrow 2}XICNOT_{1\rightarrow 2}=XX,\;CNOT_{1\rightarrow 2}IZCNOT_{1\rightarrow 2}=ZZ,\;CNOT_{1\rightarrow 2}YYCNOT_{1\rightarrow 2}=XZ\;({\rm up\; to\; phase}).
	\end{align}
	The last equation ignores the $\sqrt{-1}$-factor, since we are only considering the absolute value of $\braket{\breve{\psi}|T_{\bfa}|\breve{\psi}}$. Furthermore, we observe that if $Y$-section has odd number of $Y$'s, $\hmoment|\braket{\breve{\psi}|T_{\bfa}|\breve{\psi}}|$ should be zero because $\braket{\breve{\psi}|T_{\bfa}|\breve{\psi}}$ only outputs an imaginary value, which is also zero since $T_{\bfa}$ is Hermitian.
	Using the above arguments, we further simplify equivalent operator of $T_{\bfa}$ into, assuming only when $Y$-section is even-weighted, 
	\begin{align}
		T_{\bfa}\sim XII\ldots I\otimes XZXZ\ldots XZ\otimes ZIII\ldots \otimes IIII\ldots I\sim XII\ldots I\otimes XZII\ldots II\otimes ZIII\ldots \otimes IIII\ldots I
	\end{align}
	Conclusively, there are only $X$, $Z$, or $I$ operators. By using additional CNOT, we can further cancel $X,Z$-couples to leave only the non-trivial $4$ cases in Eq.~\eqref{eq:non_trivial_phase_strip_cases}. 
\end{proof}

From the above lemma, we shall calculate only $3$ elements in Eq.~\eqref{eq:non_trivial_phase_strip_cases}. We recall Eq.~\eqref{eq:Haar_random_l1_norm_integral} so that 
\begin{align}
	\hmoment|\braket{\breve{\psi}|ZI\ldots I|\breve{\psi}}|=\sqrt{\frac{2}{\pi 2^n}}.
\end{align}
Next, using the symmetry of ${\rm Dir}(1,1,\ldots,1)$,
\begin{align}
	\hmoment|\braket{\breve{\psi}|XI\ldots I|\breve{\psi}}|=\hmoment \sum_{\bfx\in \bfF_2}|\xi_{\bfx}||\xi_{\bfx\oplus(1,0,\ldots,0)}|=2^n\times   \int_{\triangle_{2^n}} d^{(1,1,\ldots,1)}p\left(p^{\frac{1}{2}}_1p^{\frac{1}{2}}_2\right)\simeq \frac{\pi}{4},
\end{align}
which is derived by substituting $\beta=(\frac{1}{2},\frac{1}{2},0,0,\ldots,0 )$ to  Eq.~\eqref{eq:Dirichlet_moment} and the fact that $\Gamma(M)=(M-1)!\;(M\in \mathbb{N}\cup \left\{0\right\})$. 

Lastly, 
\begin{align}
	&\hmoment|\braket{\breve{\psi}|ZXI\ldots I|\breve{\psi}}|\nonumber\\&= \hmoment\left|\sum_{\bfx'\in \bfF_2^{n-2}}\left(|\xi_{(0,0,\bfx')}\xi_{(0,1,\bfx')}|-|\xi_{(1,0,\bfx')}\xi_{(1,1,\bfx')}|+|\xi_{(0,1,\bfx')}\xi_{(0,0,\bfx')}|-|\xi_{(1,1,\bfx')}\xi_{(1,0,\bfx')}|\right)\right|\nonumber\\&=2\times  \hmoment\left|\sum_{\bfx'\in \bfF_2^{n-2}}\left(|\xi_{(0,0,\bfx')}\xi_{(0,1,\bfx')}|-|\xi_{(1,0\bfx')}\xi_{(1,1,\bfx')}|\right)\right|\nonumber\\&=2\int_{\triangle_{2^n}} d^{(1,1,\ldots,1)}p\left|\sum_{\bfx'\in \bfF_2^{n-2}}\left(p^{\frac{1}{2}}_{(0,0,\bfx')}p^{\frac{1}{2}}_{(0,1,\bfx')}-p^{\frac{1}{2}}_{(1,0,\bfx')}p^{\frac{1}{2}}_{(1,1,\bfx')}\right)\right|.
\end{align}
\begin{figure}[t]
	\includegraphics[width=12cm]{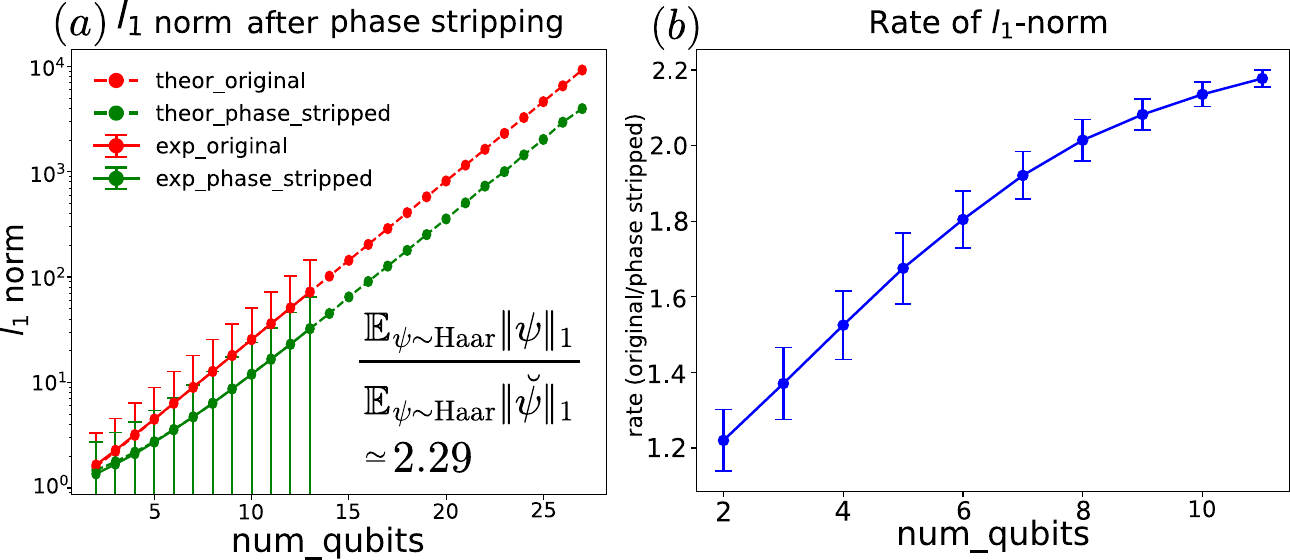}
	\caption{(a) The $l_1$-norm scaling of Haar-random pure states and their phase-stripped states. The experimental values were obtained by using $100$ random-state copies. We used Eq.~\eqref{eq:Haar_random_l1_norm_integral} to calculate the theoretical $\mathbb{E}_{\psi\sim {\rm Haar}}\|\psi\|_1$. $10000$ copies are used to estimate the theoretical $\hmoment \|\breve{\psi}\|_1$ (the right side of Eq.~\eqref{fig:random_state_l1_sup}) for each qubit size. (b) Plotting of average value of $\frac{\|\psi\|_1}{\|\breve{\psi}\|_1}$ using $500$ random-state copies.}\label{fig:random_state_l1_sup}
	
\end{figure}

In conclusion, 
\begin{align}
	\hmoment\|\breve{\psi}\|_1&\simeq \frac{1}{2^n}+\frac{2^n-1}{2^n}\sqrt{\frac{2}{\pi 2^n}}+\frac{\pi(2^n-1)}{2^{n+2}}\nonumber\\&+(4^n-2^{n+1}+1)\int_{\triangle_{2^n}} d^{(1,1,\ldots,1)}p\left|\sum_{\bfx'\in \bfF_2^{n-2}}\left(p^{\frac{1}{2}}_{(0,0,\bfx')}p^{\frac{1}{2}}_{(0,1,\bfx')}-p^{\frac{1}{2}}_{(1,0,\bfx')}p^{\frac{1}{2}}_{(1,1,\bfx')}\right)\right|\nonumber\\&\simeq (4^n-2^{n+1}+1)\int_{\triangle_{2^n}} d^{(1,1,\ldots,1)}p\left|\sum_{\bfx'\in \bfF_2^{n-2}}\left(p^{\frac{1}{2}}_{(0,0,\bfx')}p^{\frac{1}{2}}_{(0,1,\bfx')}-p^{\frac{1}{2}}_{(1,0,\bfx')}p^{\frac{1}{2}}_{(1,1,\bfx')}\right)\right|,
\end{align}
leaving dominant terms only. We can estimate the above result by sampling $p$ by Dirichlet distribution, which takes $\mathcal{O}(2^n)$-time, and take the estimator as $(4^n-2^{n+1}+1)\left|\sum_{\bfx'\in \bfF_2^{n-2}}\left(p^{\frac{1}{2}}_{(0,0,\bfx')}p^{\frac{1}{2}}_{(0,1,\bfx')}-p^{\frac{1}{2}}_{(1,0,\bfx')}p^{\frac{1}{2}}_{(1,1,\bfx')}\right)\right|$, which takes another $\mathcal{O}(2^{n-2})$ time. This offers much faster way to estimate the scale factor   $2^{-0.5n}\hmoment\|\breve{\psi}\|_1$, compared to random calculation of $\|\breve{\psi}\|_1$ where it takes $\mathcal{O}(4^n)$ for \emph{each sampled vector} $\ket{\psi}$. Fig.~\ref{fig:random_state_l1_sup} (a) shows that even when the target state is drawn Haar-randomly, $\frac{1}{2}$-FOFE still provides a constant-factor improvement. More explicitly, the $l_1$-norm of the phase-stripped version of the Haar-random state is lower than half of the original one. A caveat is that since two independent copies are needed to single-measure for real and imaginary parts, the variance bound of $\frac{1}{2}$-FOFE is reduced by a half, not a quarter. The rate of Fig.~\ref{fig:random_state_l1_sup} (a) may not directly indicate the rate between the $l_1$-norm of a \emph{fixed} state and its phase-stripped counterpart. Nevertheless, we can see from Fig.~\ref{fig:random_state_l1_sup} (b) that such a rate of $l_1$-norms for a fixed quantum state gets close to the same rate of Fig.~\ref{fig:random_state_l1_sup} (a), where the variance also decreases by increasing the number of qubits.

\section{Nonlinear DFE I: Framework}

The next two sections show the complete and algorithmic procedure for the nonlinear DFE (NLDFE) scheme, which includes a DNC-based sub-algorithm. We learned about the scheme and sampling complexity of the conventional DFE, also referred to as linear-DFE. In that scheme, the target pure state was decomposed by Pauli operators. Now, we would like to extend the spanning set to become over-complete and properly include the Pauli group, while preserving the DFE scheme to be simulated with only single-depth Pauli measurements. To do so, we recall the definition of locally scrambled diagonal of the main text. 
\begin{definition}
	The set of $n$-qubit locally-conjugated diagonal (LCD) operators is expressed as
	\begin{align}
		\mathcal{L}c_n\equiv \left\{VD(\phi)V^{\dag}\big|V\in \left\{I,H,HS\right\}^{\otimes n}, \phi:\bfF_2^n\rightarrow [0,2\pi] \right\}.
	\end{align}
\end{definition}

We use another equivalent notation $\phi=(\phi_1,\phi_2,\ldots,\phi_{2^n})\;(\forall \phi_{i(\in \bfF_2^n)}\in [0,2\pi])$ interchangeably, and denote the element as a tuple $(V,\phi)$. We note that the LCD-set properly includes the Pauli group, and hence is over-complete. It means that given a target state $\ket{\psi}$, there exists a function $f:[0,2\pi]^{2^n}\times \left\{I,H,HS\right\}^{\otimes n}\rightarrow \mathbb{C}$ such that 
\begin{align}
	\ket{\psi}\bra{\psi}=\frac{1}{2^n}\sum_{V\in \left\{I, H, HS\right\}^{\otimes n}}\int_{0}^{2\pi}d\phi_1,\ldots d\phi_{2^n} f(\phi_1,\phi_2,\ldots,\phi_{2^n};V)VD(\phi)V^{\dag}.
\end{align}
In this case, we denote $f\mapsto \ket{\psi}\bra{\psi}$.

Next, we give the generalized notion of SRE~\cite{leone2022}.
\begin{definition}
	Locally-conjugated diagonal $\alpha$-Renyi entropy ($\alpha$-LCDRE) of the pure state $\ket{\psi}$ is defined as, 
	\begin{align}
		\mathcal{L}cM_{\alpha}(\psi)\equiv \min_{f\mapsto \ket{\psi}\bra{\psi}}\mathcal{L}cM_{\alpha}(\psi,f),
	\end{align}
	where $\|f\|_{2\alpha}\equiv \frac{1}{2^n}\left(\sum_{V\in \left\{I, H, HS\right\}^{\otimes n}}\int_{0}^{2\pi}|f(\phi;V)|^{2\alpha}d\phi\right)^{\frac{1}{2\alpha}}$, and 
	\begin{align}\label{eq:ltm_renyi}
		\mathcal{L}cM_{\alpha}(\psi,f)&\equiv\frac{1}{1-\alpha}\log_2\left(\frac{1}{2^{\alpha n}}\sum_{V\in \left\{I, H, HS\right\}^{\otimes n}}\int_{0}^{2\pi}d\phi_1,\ldots d\phi_{2^n}|f|^{2\alpha}\right)-n\nonumber\\&=H_\alpha(\widetilde{f})+\frac{2\alpha}{1-\alpha}\log_2\|f\|_2+\frac{2\alpha-1}{1-\alpha}n,
	\end{align}
	where $H_{\alpha}$ is Renyi entropy, $H_{\alpha}(p)\equiv\frac{1}{1-\alpha}\log_2\left(\sum_{\bfy}\int d\bfx \;p^{\alpha}(\bfx,\bfy)\right)$, and $\widetilde{f}\equiv \frac{|f|^2}{4^n\|f\|_2^2}$ is a normalized distribution.
\end{definition}

We describe how to generalize the DFE scheme with this LCD-set. The point is that we can still estimate the expectation value $\braket{(V,\phi)\in \mathcal{L}c_n}$ via only a single-depth Pauli measurement. Indeed, 
\begin{align}\label{eq:non_linear_estimator}
	\braket{(V,\phi)}=\tr{\rho (V,\phi)}= \tr{V^{\dag}\rho V D(\phi)}= \sum_{\bfx}\tr{V^{\dag}\rho Ve^{i\phi (\bfx)}\ket{\bfx}\bra{\bfx}}=\sum_{\bfx}\braket{\bfx|V^{\dag}\rho V|\bfx}e^{i\phi(\bfx)}. 
\end{align}
Therefore, we first twirl $\rho$ by the single qubit Clifford operations $V^{\dag}$ and then we measure in the computational basis to get the outcome $\bfx\in \bfF_2^n$ followed by the estimator $e^{i\phi(\bfx)}$. Compared to the linear DFE, this estimator is not linear since $\phi$ is an arbitrary mapping. 
Therefore, given that $f\mapsto \ket{\psi}\bra{\psi}$,
we can make the generalized fidelity estimation, a \emph{nonlinear} $\alpha$-DFE ($\alpha$-NLDFE) scheme by sampling $(\phi,V)$ from the probability distribution $
\left\{\frac{|f(\phi_1,\;\ldots,\phi_{2^n},V)|^{2\alpha}d\phi_1d\phi_2\ldots d\phi_{2^n}}{2^{2\alpha n}\|f\|_{2\alpha}^{2\alpha}}\right\}$, twirl $\rho$ by $V^{\dag}$ and measure to obtain $\bfx$, and then we take the estimator 

\begin{align}\label{eq:exact_est_non_linear_dfe}
	\widehat{\braket{\psi|\rho|\psi}}=2^{(2\alpha-1)n}\|f\|_{2\alpha}^{2\alpha}|f(\phi_1,\;\ldots,\phi_{2^n},V)|^{-2\alpha+1}\cos(\phi(\bfx)+\arg(f(\phi_1,\;\ldots,\phi_{2^n},V))),
\end{align} 
such that $f=|f|e^{i\arg(f)}$, following Eq.~\eqref{eq:non_linear_estimator}. The cosine term is obtained because we only need to consider the real part of the estimator. Using $\cos^2\le 1$, we can find the upper bound of the estimation variance whose derivation is similar to that of Eq.~\eqref{eq:MSE_upper}, 
\begin{align}
	{\rm Var}(\rho,\psi,\alpha)+\braket{\psi|\rho|\psi}^2&= \mathbb{E}\left(\widehat{ \braket{\psi|\rho|\psi}}^2\right)\nonumber\\&\le2^{(2\alpha-2)n}\|f\|_{2\alpha}^{2\alpha}\sum_{V\in \left\{I, H, HS\right\}^{\otimes n}}\int_{0}^{2\pi}d\phi_1,\ldots d\phi_{2^n} \left|f(\phi_1,\phi_2,\ldots,\phi_{2^n};V)\right|^{-2\alpha+2}\sum_{\bfx\in \bfF_2^n}\braket{\bfx|V^{\dag}\rho V|\bfx}\nonumber\\&\le 2^{(1-\alpha)\mathcal{L}cM_\alpha(\psi,f)+\alpha\mathcal{L}cM_{1-\alpha}(\psi,f)}.
\end{align}
We again ignore the term $\braket{\psi|\rho|\psi}^2\le 1$. To find the optimal $\alpha$ that gives a minimum of the right side value, we again use the monotonicity of $\alpha$-Renyi entropy. That is for a fixed $f\mapsto\ket{\psi}\bra{\psi}$, 
\begin{align}
	(1-\alpha)\mathcal{L}cM_\alpha(\psi)+\alpha\mathcal{L}cM_{1-\alpha}(\psi)&=(1-\alpha)H_{\alpha}(\widetilde{f})+\alpha H_{1-\alpha}(\widetilde{f})+2\log_2\|f\|_2+n,
\end{align}
which heats the minimum at $\alpha=\frac{1}{2}$ by Cor.~\ref{Cor:renyi_convex}.
Conclusively, 
\begin{align}
	{\rm Var}(\rho,\psi,\alpha)\le \min_{f\mapsto \ket{\psi}\bra{\psi}}\left\{2^{\mathcal{L}cM_{\frac{1}{2}}(\psi,f)}\right\}=2^{\mathcal{L}cM_{\frac{1}{2}}(\psi,f^{\ast})}=\left(\frac{1}{2^n}\sum_{V\in \left\{I, H, HS\right\}^{\otimes n}}\int_{0}^{2\pi}d\phi_1,\ldots d\phi_{2^n}|f^{\ast}_{\frac{1}{2}}|\right)^2\nonumber&\equiv\left\|f^{\ast}_{\frac{1}{2}}\right\|_1^2,
\end{align}
where $f^{\ast}_\frac{1}{2}\equiv {\rm argmin}_{f\mapsto\ket{\psi}\bra{\psi}}\left\{\mathcal{L}cM_{\frac{1}{2}}(\psi,f)\right\}$.

We end this section with two important remarks. First, 
from Eq.~\eqref{eq:exact_est_non_linear_dfe}, we should note that exact variance term includes $\cos^2$-term, which we ignored as bounding with unity. It means that if we make $f\mapsto \ket{\psi}\bra{\psi}$ with the complex-valued function $f$, we could gain additional scale-factor improvement of the estimation variance. However, keeping the cosine term would make the computation of variance much more complex. Therefore, we leave the exact calculation of the variance as future work. Second, finding $f^{\ast}_{\frac{1}{2^n}}$ is nearly impossible because the LCD-set is infinitely large and over-complete. At least, we could find expansion showing better sampling complexity compared to the normal Pauli expansion, so in the following section, we objects to make the sub-optimal algorithm which find $f$ such that $2^{\mathcal{L}cM_{\frac{1}{2}}(\psi,f)}\le2^{M_\frac{1}{2}(\psi)}=\|\psi\|_1^2$. 

\section{Nonlinear DFE II: Divide-and-conquer (DNC)-based algorithm}

In the previous section, we introduced the general formalism of nonlinear DFE. It is challenging to design an algorithm that achieves the optimal sampling complexity for nonlinear DFE. Instead, let us develop a sub-optimal method that still presents much better sampling overhead compared to the original linear DFE. 

To do so, we first define qubit-wise commuting (QWC)~\cite{julia2025} Pauli subgroup that is an Abelian subgroup equivalent to Pauli $Z$-subgroup under local Clifford operations. Let us say $S$ is an $n$-qubit QWC-subgroup corresponding to the local Clifford operation $V$. We also denote $\qwcgroup$ as the set of possible $n$-qubit QWC-subgroups. Then we can estimate the expectation value of the linear combination of elements that is $\braket{\sum_{P\in S}c_P P}$, following the identity below, 
\begin{align}
	\braket{\sum_{P\in S}c_P P}=\sum_{P\in S} c_P \tr{\rho P}=\sum_{\bfa\in \bfF_2^n}c_{VZ^{\bfa}V^{\dag}}\tr{V\rho V^{\dag}Z^{\bfa}}&=\sum_{\bfa\in \bfF_2^n}c_{VZ^{\bfa}V^{\dag}}\sum_{\bfb\in \bfF_2^n}(-1)^{\bfa\cdot \bfb}\braket{\bfb|V\rho V^{\dag}|\bfb}\nonumber\\&=\sum_{\bfb\in \bfF_2^n}\left(\sum_{\bfa\in \bfF_2^n}c_{VZ^{\bfa}V^{\dag}}(-1)^{\bfa\cdot \bfb}\right)\braket{\bfb|V\rho V^{\dag}|\bfb}.
\end{align}
Therefore, estimation scheme measures $\rho$ with the computational basis to obtain the outcome $\bfb$ then takes the estimator as $\sum_{\bfa\in \bfF_2^n}c_{VZ^{\bfa}V^{\dag}}(-1)^{\bfa\cdot \bfb}$. 

More importantly, we note that the Pauli group can be the union of $\mathcal{O}(3^n)$-number of QWC groups. One remark is that if we relax the condition of a subgroup by allowing an entangled Clifford operation $V$, then the Pauli group can be partitioned into $2^n$-number of subgroups. Let us consider QWC-restricted cases only.  

Following the above argument, we further improve the required sampling copies (estimation variance) by the following routines which are $l_1$-sampling version of Pauli grouping~\cite{julia2025}: We first divide the elements of the Pauli group into several QWC subgroups, say some elements into one QWC subgroup $S$ form the partial linear combination of the target state $\ket{\psi}\bra{\psi}$ as $\sum_{\bfa\in \bfF_2^n}c_{V^{(S)}Z^{\bfa}V^{(S)\dag}} V^{(S)}Z^{\bfa}V^{(S)\dag} (\ne \ket{\psi}\bra{\psi})$. Then we estimate each $\braket{\sum_{P\in S}c_P P (\ne \ket{\psi}\bra{\psi})}$ of input state $\rho$ for only QWC subgroups which contain non-trivial Pauli supports of $\ket{\psi}\bra{\psi}$, taking the sum at the last. By doing so, we attain further improvement of sampling copies compared to when we just take the estimator $\|\psi\|_1$ or $-\|\psi\|_1$.  

We note that such a QWC-based method keeps the linear post-processing ($\bfb\rightarrow \bfa\cdot \bfb$ for assigned $\bfa$'s). Now, we further generalize and improve such a scheme by using non-linear post-processing of the obtained outcome. We start with one QWC group $S$ containing a non-trivial Pauli support $\sum_{\bfa\in \bfF_2^n}c_{V^{(S)}Z^{\bfa}V^{(S)\dag}} V^{(S)}Z^{\bfa}V^{(S)\dag}$. This is already decomposed with the Pauli operators. Then how about decomposing with the LCD-set components of a given twirling $V^{(S)}$? That is, we would like to do that
\begin{align}
	&\sum_{\bfa\in \bfF_2^n}c_{V^{(S)}Z^{\bfa}V^{(S)\dag}} V^{(S)}T_{\bfa}V^{(S)\dag}=\frac{1}{2^n}\int_{0}^{2\pi}d\phi_1,\ldots d\phi_{2^n} f_S(\phi_1,\phi_2,\ldots,\phi_{2^n};V^{(S)})V^{(S)}D(\phi)V^{(S)\dag}\nonumber\\&\Rightarrow \sum_{\bfa\in \bfF_2^n}c_{V^{(S)}Z^{\bfa}V^{(S)\dag}} T_{\bfa}=\frac{1}{2^n}\int_{0}^{2\pi}d\phi_1,\ldots d\phi_{2^n} f_S(\phi_1,\phi_2,\ldots,\phi_{2^n};V^{(S)})D(\phi)\nonumber\\&\Rightarrow\;\forall \bfb\in \bfF_2^n,\; 2^n c_{V^{(S)}Z^{\bfb}V^{(S)\dag}}=\frac{1}{2^n}\int_{0}^{2\pi}d\phi_1,\ldots d\phi_{2^n} f_S(\phi_1,\phi_2,\ldots,\phi_{2^n};V^{(S)}\sum_{\bfx\in \bfF_2^n}(-1)^{\bfb\cdot \bfx}e^{i\phi(\bfx)}\nonumber\\&=\frac{1}{2^n}\int_{0}^{2\pi}d\phi_1,\ldots d\phi_{2^n} f_S(\phi_1,\phi_2,\ldots,\phi_{2^n};V^{(S)})\widehat{e^{i\phi}}_{\bfb}
\end{align}
We mark the subscript $S$ of $f$ to indicate that $f_S$ is not the whole coefficients for $\ket{\psi}\bra{\psi}$. The third equality is derived by taking the product $\bfF^{\bfb}$ and trace. We also denoted the Walsh-Hadamard transform (WHT) as $\widehat{f}_\bfb\equiv \sum_{\bfx\in \bfF_2^n}f(\bfx)(-1)^{\bfb\cdot \bfx}$ whose inverse transform is $\widehat{f}^{-1}_\bfb\equiv \frac{1}{2^n}\sum_{\bfx\in \bfF_2^n}f(\bfx)(-1)^{\bfb\cdot \bfx}$. Next, we take the inverse transform on both sides. Then we finally get that 
\begin{align}
	2^n\widehat{c^{(S)}}^{-1}_\bfb=\widehat{c^{(S)}}_\bfb=\frac{1}{2^n}\int_{0}^{2\pi}d\phi_1,\ldots d\phi_{2^n} f_S(\phi_1,\phi_2,\ldots,\phi_{2^n};V^{(S)})e^{i\phi(\bfb)}, 
\end{align}
where we used the notation $c^{(S)}_{\bfa}=c_{V^{(S)}Z^{\bfa}V^{(S)\dag}}$
We know the left side because $c$ is already given by $\ket{\psi}\bra{\psi}$. The problem shrinks to a sub-optimal one where we minimize $\|f^{(S)}\|_1=\frac{1}{2^n}\int_{0}^{2\pi}d\phi_1,\ldots d\phi_{2^n} |f_S(\phi_1,\phi_2,\ldots,\phi_{2^n};V^{(S)})|$. We prove that this problem can be deterministically solved by using a well-known correspondence between the atomic norm ~\cite{yang2015} and the infinity norm over the hypercube. We give the main statement and proof for completeness. 
\begin{lemma}\label{lem:optimal_qwc_nonlinear}
	Suppose for a given $V\in {\rm Cl}_1^{\otimes n}$, then we obtain that given $\bfc\in \mathbb{R}^{2^n}\backslash\left\{\mathbf{0}\right\}$,
	\begin{align}
		\inf\left\{\|f\|_1\bigg|\forall \bfb\in \bfF_2^n,\;c_\bfb=\int_{0}^{2\pi}d\phi_1,\ldots d\phi_{2^n} f(\phi_1,\phi_2,\ldots,\phi_{2^n};V)e^{i\phi(\bfb)}\right\}=\|c\|_{\infty}\le \|c\|_1,
	\end{align}
	where the infinite norm (a.k.a atomic norm) is defined as $\|c\|_{\infty}\equiv \max_{\bfa\in\bfF_2^n}\left\{|c_{\bfa}|\right\}$. An optimal function $f^*$ of minimal $l_1$-norm exists so that all $e^{i\phi}$ becomes a real sign function.
\end{lemma}

\begin{proof}

	
	We first note that $\forall |e^{i\phi(\bfb)}|=1$. If $f(\phi_1,\ldots,\phi_{2^n},V)$ is a complex value, we can let the corresponding vector $e^{i\phi}$ absorb the phase of $f$, then sum (merge) the terms with the same phase function ($\because$  merging the coefficients with the same vector always gives a better or equal $l_1$ norm). Hence, we assume $f$ is real. Furthermore, let $A$ be the integral transform $A\bfy(\phi)\equiv \int_{[0,2\pi]^{2^n}}d\phi e^{i\phi(\bfb)}y(\phi)$. We know that from a general form of $l_1$-optimization, while assuming the feasible space reality, 
	\begin{align}
		\min_{Ay=c}\|y\|_1=\min_{\substack{Re(A)y=c,\\Im(A)y=\mathbf{0}}}\|y\|_1\ge \min_{Re(A)y=c}\|y\|_1,
	\end{align} 
	we also note that it suffices to find the optimal decomposing vector, which is real, given that the solution of the third super-case exists. Since $y$ is a real-valued function, we obtain that
	\begin{align}
		{\rm Re}(A)y-c=\sum_{\bfb\in \bfF_2^n}\bm{e}_\bfb\left(\int_{[0,2\pi]^{2^n}}d\phi{\rm Re(e^{i\phi (\bfb)})y(\phi)}-c_{\bfb}\right)=\sum_{\bfb\in \bfF_2^n}\bm{e}_\bfb\left(\int_{[0,2\pi]^{2^n}}d\phi\cos(\phi (\bfb))y(\phi)-c_{\bfb}\right).
	\end{align}
	Here, $\left\{\bm{e}_{\bfb}\right\}_{\bfb\in \bfF_2^n}$ is the computational bases. We note that $\left\{\sum_{\bfb}e_{\bfb}\cos(\phi(\bfb))\big|\phi\in [0,2\pi]^{^{2^n}}\right\}$ is exactly the hypercube $[-1,1]^{2^n}$. Since $l_1$-norm is a convex function, an optimal solution happens when we decompose with extreme points of the hypercube, which is $\left\{1,-1\right\}^{2^n}$.
	
	Therefore, the problem shrinks to finding $A$ when the hypercube $\square$ that is,
	\begin{align}
		\square\equiv A\times {\rm conv}\left\{\eta^{(B)}|\forall \bfx\in \bfF_2^n, \eta^{(B)}_\bfx=(-1)^{B(\bfx)},B\in \mathcal{B}_{2^n}\right\}
	\end{align}
	where $\mathcal{B}_{2^n}$ denotes a set of possible Boolean functions $B:\bfF_2^n\rightarrow \bfF_2$ touches vector $c$. We note that if we divide $\|c\|_{\infty}$ to $c$, all elements has the modulus lower or equal than $1$. More importantly, there exists at least one elements of unit modulus. 
	
	Next, we claim that there exists $f^{\ast}$ such that it satisfies the condition of decomposing $c$ and achieves $\|f\|_1=\|c\|_{\infty}$. If we prove that, we conclude that this is optimal. Because, if the optimal $A$ is lower than  $\|c\|_{\infty}$, $A\times {\rm conv}\left\{{\rm diag}(e^{i\phi(\bfx)})_{\bfx\in \bfF_2^n}|\phi\in [0,2\pi]^{n}\right\}$ cannot touch $c$ since no any convex combinations of vectors $e^{i\phi}$ reach the modulus $\|c\|_{\infty}$, which is a contradiction. Now, let us construct $f^{\ast}$. $\forall \bfb\in \bfF_2^n,$
	
	\begin{align}
		c_{\bfb}=\|c\|_{\infty}\frac{c_{\bfb}}{\|c\|_\infty}=\|c\|_{\infty}\left(\frac{1+\frac{c_\bfb}{\|c\|_{\infty}}}{2}+(-1)\frac{1-\frac{c_\bfb}{\|c\|_{\infty}}}{2}\right).
	\end{align}
	Therefore, we conclude that
	\begin{align}
		c&=\sum_{\bfb\in \bfF_2^n}\|c\|_{\infty}\left(\frac{1+\frac{c_\bfb}{\|c\|_{\infty}}}{2}+(-1)\frac{1-\frac{c_\bfb}{\|c\|_{\infty}}}{2}\right)=\sum_{\bfb\in \bfF_2^n}\|c\|_{\infty}\sum_{a\in \bfF_2}\frac{1+(-1)^{a}\frac{c_\bfb}{\|c\|_{\infty}}}{2}(-1)^a\nonumber\\&=\|c\|_{\infty}\sum_{B\in \mathcal{B}_{2^n}}\prod_{\bfa\in \bfF_2^n}\frac{1+(-1)^{B(\bfa)}\frac{c_\bfa}{\|c\|_{\infty}}}{2}\eta^{(B)}.
	\end{align}
	We note that $\eta^{(B)}$ corresponds to multiple controlled-$Z$ gate operations.
	
	Conclusively, we take the functional value $f^{\ast}$ as  $\|c\|_{\infty}\prod_{\bfa\in \bfF_2^n}\frac{1+(-1)^{B(\bfa)}\frac{c_\bfa}{\|c\|_{\infty}}}{2}$ for each Boolean function argument $B$, and zero for otherwise. The $f^{\ast}$ is non-negative, and sums to $\|c\|_{\infty}$. Indeed, 
	\begin{align}
		\sum_{B\in \mathcal{B}_{2^n}}\prod_{\bfa\in \bfF_2^n}\frac{1+(-1)^{B(\bfa)}\frac{c_\bfa}{\|c\|_{\infty}}}{2}=\prod_{\bfa\in \bfF_2^n}\left(\frac{1+\frac{c_\bfa}{\|c\|_{\infty}}}{2}+\frac{1-\frac{c_\bfa}{\|c\|_{\infty}}}{2}\right)=1.
	\end{align}
	The proof is completed.
\end{proof}

Let us further clarify the solution structure and a compressed NLDFE routine by the DNC algorithm. We refer the $l_1$-norm of $f^{(S)*}$ as 
\begin{align}
	\|f^{(S)}\|_1\equiv \frac{1}{2^n}\int_{[0,2\pi]^{2^n}}d\phi |f^{(S)*}(\phi,V^{(S)})|,
\end{align}
implying that $\|f\|_1=\sum_{V\in \left\{I,H,HS\right\}^{\otimes n}}\|f^{(S)*}\|_1$. 
From Lem.~\ref{lem:optimal_qwc_nonlinear}, we saw that non-trivial value of $f^{(S)*}$ happens only on the set of Boolean function $\mathcal{B}_n\simeq \bfF_2^{2^n}$, and the result was,
\begin{align}\label{eq:app_optimal_wht_solution}
	f^{(S)*}(B)=\prod_{\bfa\in \bfF_2^n}\left(\frac{1+(-1)^{B(\bfa)}\frac{\widehat{c^{(S)}}_{\bfa}}{\|\widehat{c^{(S)}}\|_{\infty}}}{2}\right).
\end{align}
Importantly, $\sum_{B\in \mathcal{B}_n}f^{(S)*}(B)=1$, and $\forall B\in \mathcal{B}_n,\;f^{(S)*}\ge 0$. That is, $f^{(S)*}$ forms an probability distribution over $\mathcal{B}_n$. Therefore, the Boolean function sampling routine would be as follows: we sample $S\in \qwcgroup$ by the distribution $\left\{\frac{\|\widehat{c^{(S)}}\|_{\infty}}{\sum_{S\in \qwcgroup}\|\widehat{c^{(S)}}\|_{\infty}}\right\}_{S\in \qwcgroup}$ after possible distribution of Pauli coefficients into the QWC-subgroups. Then we sample $B\in \mathcal{B}_n$ following the probability distribution $f^{(S)*}$ shown as Eq.~\eqref{eq:app_optimal_wht_solution}.

However, we can simplify the sampling routine by combining the Pauli grouping method~\cite{julia2025}. Within a single QWC-subgroup $S$, after sampling $B$, we remember that we measured $V^{\dag}\rho V$ in the computational bases to obtain the outcome $\bfb\in \bfF_2^n$ and finally the estinator $(-1)^{B(\bfb)}$. However, we note that all $V^{(S)}D(\phi)V^{(S)\dag}$'s commute each other within the same $S$. It implies that we can use the measurement outcome $\bfb$ to estimate the whole parts within $S$ by taking the estimator as $\sum_{B\in \mathcal{B}_n}f^{(S)*}(B)(-1)^{B(\bfb)}$. Furthermore, by using Eq.~\eqref{eq:app_optimal_wht_solution}, we get that 
\begin{align}
	\sum_{B\in \mathcal{B}_n}f^{(S)*}(B)(-1)^{B(\bfb)}&=	\sum_{B\in \mathcal{B}_n}\prod_{\bfa\in \bfF_2^n}\left(\frac{1+(-1)^{B(\bfa)}\frac{\widehat{c^{(S)}}_{\bfa}}{\|\widehat{c^{(S)}}\|_{\infty}}}{2}\right)(-1)^{B(\bfb)}\nonumber\\&=\sum_{c=0,1}\left(\frac{(-1)^c+\frac{\widehat{c^{(S)}}_{\bfb}}{\|\widehat{c^{(S)}}\|_{\infty}}}{2}\right)\times \underbrace{\sum_{B\in \mathcal{B},B(\bfb)=c}\prod_{\bfa\in \bfF_2^n,\bfa\ne \bfb}\left(\frac{1+(-1)^{B(\bfa)}\frac{\widehat{c^{(S)}}_{\bfa}}{\|\widehat{c^{(S)}}\|_{\infty}}}{2}\right)}_{=1}\nonumber\\&=\frac{\widehat{c^{(S)}}_{\bfb}}{\|\widehat{c}^{(S)}\|_{\infty}}.
\end{align}
In conclusion, sampling $B$ following $f^{(S)*}$ is unnecessary in this DNC-based algorithm. After we sample $S$, we take the estimator as $\sum_{S\in \qwcgroup}\|\widehat{c^{(S)}}\|_{\infty}\times \frac{\widehat{c^{(S)}}_\bfb}{\|\widehat{c^{(S)}}\|_{\infty}}$. Hence, by Hoeffding inequality, the required sampling complexity depends on the square of  $\sum_{S\in \qwcgroup}\|\widehat{c^{(S)}}\|_{\infty}$. Furthermore, the way of distributing Pauli coefficients over the QWC-groups is non-unique, so that we could take the minimum among the distributions. 
Finally, we conclude our DNC-NLDFE statement as a theorem:
\begin{theorem}\label{thm:NLDFE_sampling_complexity}
	Using Pauli measurements only, the required sampling copies $N$ for estimating the fidelity with the target state $\ket{\psi}$ within $\epsilon$ additive error and $\delta_f$ failure probability is described as $N=\mathcal{O}\left(\frac{\min\left\{\sum_{S\in \qwcgroup}\|\widehat{c^{(S)}}\|_{\infty}\right\}^2}{\epsilon^2}\log(\delta_f^{-1})\right)$. Here, the minimum is taken over all partitions of Pauli coefficients into the set of QWC-subgroups.
\end{theorem}

\begin{figure}[t]
	\includegraphics[width=10cm]{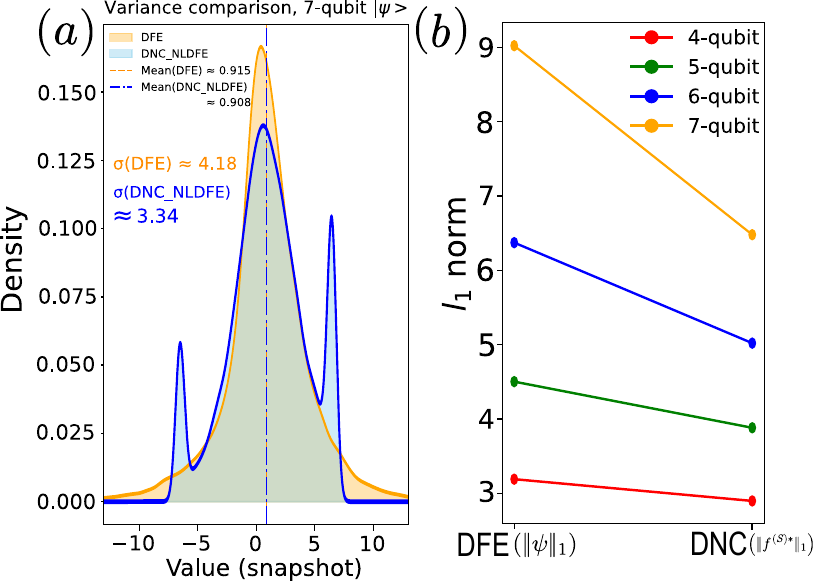}
	\caption{(a) Comparison of the estimation variance between DFE~\cite{julia2025} over QWC-groups and DNC-based NLDFE for a randomly chosen target pure state and input state. (b) Improvement in $l_1$-norm of $100$-copies of Haar random pure states \emph{via} the DNC-based algorithm.}\label{fig:DnC-based_nldfe_variance}
	
\end{figure}

Our variance bound is guaranteed to be improved over the previous one of linear-DFE by the following inequality,  
\begin{align}
	\|\widehat{c^{(S)}}\|_\infty=\max_{\bfb\in \bfF_2^n}\left\{\left|\sum_{\bfa\in \bfF_2^n,c^{(S)}_\bfa\ne 0}c^{(S)}_{\bfa}(-1)^{\bfa\cdot \bfb}\right|\right\}\le \sum_{\bfa\in \bfF_2^n,c^{(S)}_\bfa\ne 0}|c^{(S)}_{\bfa}|=\|c^{(S)}\|_1.
\end{align}

In summary, our scheme is organized as follows. We divide the Pauli coefficients of $\ket{\psi}\bra{\psi}$ into several QWC-subgroups. Then for each QWC-subgroup $S\in \qwcgroup$ containing the portion $c^{(S)}$ of Pauli coefficients. We then sample $S\in \qwcgroup$ from the distribution $\left\{\frac{\|\widehat{c^{(S)}}\|_{\infty}}{\sum_{S\in \qwcgroup}\|\widehat{c^{(S)}}\|_{\infty}}\right\}_{S\in \qwcgroup}$. Next, we measure $V^{(S)\dag}\rho V^{(S)}$ with the computational basis to obtain the outcome $\bfb$, which reads the estimator $\sum_{S\in \qwcgroup}\|\widehat{c^{(S)}}\|_{\infty}\times \frac{\widehat{c^{(S)}}_\bfb}{\|\widehat{c^{(S)}}\|_{\infty}}$. The whole process does not necessitate sampling the Boolean function from $f^{(S)\ast}$. There is an important caveat. Suppose we operate similar optimization scheme following $S_1\rightarrow S_2\rightarrow \ldots\rightarrow S_{3^n}\in \qwcgroup$. When we move to the next QWC group, say $S_{i(\in [3^n])}$ we need to set the coefficients of the duplicate Pauli support with the previously completed QWC-groups$[V^{(S_i)}Z^{\bfa}V^{(S_i)}=V^{(S_j)}Z^{\bfa}V^{(S_j)}\;(\forall j\in [i-1])]$ to be zero. Fig.~\ref{fig:DnC-based_nldfe_variance} demonstrates the reduction of estimation variance and $l_1$-norm of the coefficient $f=(f^{(S)\ast})_{S\in \qwcgroup}$ compared to that of the original DFE in the Haar-random case. The reduction amount becomes larger as we increase the system size.

Now, we have different norm scaling as the infinity norm of WH-spectra, which is properly lower than the $l_1$-norm in typical target states. For the stabilizer target state, our method gives the same value as the conventional bound. 

\begin{proposition}
	If the target pure state $\ket{\psi}$ is a stabilzer state, then $\|\widehat{c^{(S)}}\|_{\infty}=\|c^{(S)}\|_1$, hence $\min\left\{\sum_{S\in \qwcgroup}\|\widehat{c^{(S)}}\|_{\infty}\right\}$ is $1$. 
\end{proposition}

\begin{proof}
	First, we should note that non-trivial support of $c^{(S)}$ always lies on some subspace of $\bfF_2^n$. This is because if both $V^{(S)}Z^{\bfa_1}V^{(S)\dag}$ and $V^{(S)}Z^{\bfa_1}V^{(S)\dag}$ belong to the stabilzer group for $\psi$, so does $V^{(S)}Z^{\bfa_1+\bfa_2}V^{(S)\dag}$. We denote such a subspace as $L^{(S)}$ whose orthonormal (with respect to binary inner product) basis is $\mathcal{B}_S\equiv \left\{\bfv_1,\bfv_2,\ldots,\bfv_{\dim (L^{(V)})}\right\}$.
	
	Now, suppose that the stabilizer state $\ket{\psi}$ has the coefficient $\frac{(-1)^{p_{\bfv}}}{2^n}$ for $V^{S}Z^{\bfv}V^{(S)\dag}\; (\bfv\in \mathcal{B}_S)$. Then $\bfc^{(S)}_{\bfa}=\frac{1}{2^n}(-1)^{\sum_{\bfv\in \mathcal{B}_S}p_\bfv(\bfv\cdot \bfa)}$. Conclusively, 
	\begin{align}
		\|\widehat{c^{(S)}}\|_\infty=\max_{\bfb\in \bfF_2^n}\left\{\left|\sum_{\bfa\in \bfF_2^n,c^{(S)}_\bfa\ne 0}\frac{1}{2^n}(-1)^{\sum_{\bfv\in \mathcal{B}_S}p_\bfv(\bfv\cdot \bfa)}(-1)^{\bfa\cdot \bfb}\right|\right\}&=\max_{\bfb\in \bfF_2^n}\left\{\left|\sum_{\bfa\in \bfF_2^n,c^{(S)}_\bfa\ne 0}\frac{1}{2^n}(-1)^{\bfa\cdot \left(\bfb+\sum_{\bfv\in \mathcal{B}_S}p_\bfv \bfv\right)}\right|\right\}\nonumber\\&=\sum_{\bfa\in \bfF_2^n,c^{(S)}_\bfa\ne 0}|c^{(S)}_{\bfa}|\nonumber\\&=\|c^{(S)}\|_1,
	\end{align}
	where the optimum holds if $\bfb=\sum_{\bfv\in \mathcal{B}_V}p_\bfv \bfv$. The remaining proof is deducted straightforwardly. 
\end{proof}

Before ending this section, we analyse the time complexity of the DNC algorithm. Computing Pauli coefficients takes $\mathcal{O}(4^n)$-time, which is also the time for the original DFE. We make another $\mathcal{O}(4^n)$-sized memory, say the \emph{reference memory}, to record whether a given Pauli support has non-zero coefficients (e.g., $0$ or $1$). For each QWC-subgroup, we need the inverse Walsh-Hadamard transform of the coefficients, which takes $\mathcal{O}(n2^n)$-time~\cite{scheibler2015}. Before the WH-transform, whether we adopt each coefficient as itself or zero is determined by consulting the reference memory, which takes $\mathcal{O}(\log(4^n))=\mathcal{O}(n)$-time. When we use the coefficient as itself, we flip the corresponding element in the reference memory to zero. Since WH-transform is processed over $\mathcal{O}(3^n)$ number of QWC-groups, the total time complexity is $\mathcal{O}(n6^n)$. We see that the worst-case time complexity to run our scheme is exponentially large in the number of qubits. Nevertheless, it provides a constructive method for simulating nonlinear DFE and significantly improves sampling complexity compared to linear DFE, which also takes exponentially many time by $n$.

\section{Appendix A: Properties of incomplete beta functions}
Here, we prove Eq.~\eqref{eq:properties_incomplete_beta}.
We first recall the definition of incomplete beta function,
\begin{align}
	B\left(x;a,b\right)
	\equiv \int_0^{x}t^{a-1}(1-t)^{b-1}dt,\; B(1;a,b)\equiv B(a,b)=\frac{\Gamma(a)\Gamma(b)}{\Gamma(a+b)}
\end{align}

Then 
\begin{align}\label{eq:beta1}
	B(a+1,a)=\frac{\Gamma(a+1)\Gamma(a)}{\Gamma(2a+1)}=\frac{\Gamma(a)^2a}{\Gamma(2a)2a}=\frac{1}{2}B(a,a),
\end{align}
where we used the property $\Gamma(a+1)=a\Gamma(a)$. Since, the beta function is symmetric, we also note $B(a,a+1)=\frac{1}{2}B(a,a)$. Next, using variable change $z=1-t$,
\begin{align}\label{eq:beta2}
	B\left(\frac{1}{2};a,b\right)=\int_0^{\frac{1}{2}}t^{a-1}(1-t)^{b-1}dt=\frac{1}{2}\left(\int_0^{\frac{1}{2}}t^{a-1}(1-t)^{b-1}dt-\int_1^{\frac{1}{2}}(1-z)^{a-1}z^{b-1}dz\right)&=\frac{1}{2}\int_0^1t^{a-1}(1-t)^{b-1}dt\nonumber\\&=\frac{1}{2}B(a,b).
\end{align}
Lastly, using Eqs.~\eqref{eq:beta1} and~\eqref{eq:beta2},
\begin{align}
	B\left(\frac{1}{2};a+1,a\right)=\int_0^{\frac{1}{2}}t^{a}(1-t)^{a-1}dt&=\left[-\frac{t^a}{a}(1-t)^{a}\right]^{\frac{1}{2}}_0+\int_0^{\frac{1}{2}}t^{a-1}(1-t)^{a}dt=-\frac{1}{a2^{2a}}+ B\left(\frac{1}{2};a,a+1\right)\nonumber\\&=-\frac{1}{a2^{2a}}+\frac{1}{2}B(a,a+1)=\frac{1}{4}B(a,a)-\frac{1}{a2^{2a}}.
\end{align}

\bibliographystyle{apsrev4-1}
\bibliography{apssamp}

\end{document}